\documentclass[12pt]{amsart}  

\usepackage{graphicx}
\usepackage{here} 
\usepackage{array} 

\usepackage{vmargin}
\usepackage{amssymb}
\usepackage{mathrsfs}
\usepackage[all]{xy}
\usepackage[usenames,dvipsnames]{color}
\usepackage{soul}
\RequirePackage[colorlinks,linkcolor=blue,citecolor=LimeGreen,urlcolor=red]{hyperref} 
\usepackage{amsmath}

\newtheorem{theorem}{Theorem}[section]

\newtheorem{cor}[theorem]{Corollary}

\newtheorem{lemma}[theorem]{Lemma}

\newtheorem{prop}[theorem]{Proposition}
\newtheorem{remark}[theorem]{Remark}

\numberwithin{equation}{section}

\newcommand{\R}{\mathbb{R}}

\newcommand{\N}{\mathbb{N}}

\newcommand{\T}{\mathbb{T}}

\newcommand{\abs}[1]{\left|#1\right|}
\newcommand{\eps}{\varepsilon}
\newcommand{\norm}[1]{\left\|#1\right\|}

\renewcommand{\leq}{\leqslant}
\renewcommand{\geq}{\geqslant}
\renewcommand{\bar}{\overline}
\renewcommand{\tilde}{\widetilde}
\newcommand{\pa}[1]{\left(#1\right)}
\newcommand{\cro}[1]{\left[#1\right]}
\newcommand{\br}[1]{\left\{#1\right\}}
\newcommand\restr[2]{{
  \left.\kern-\nulldelimiterspace 
  #1 
  \right|_{ #2} 
  }}

\newcommand{\eu}{\mathrm{e}}
\newcommand{\dual}[2]{\langle #1, #2\rangle}

\makeatletter
\def\namedlabel#1#2{\begingroup
    #2%
    \def\@currentlabel{#2}%
    \phantomsection\label{#1}\endgroup
}
\makeatother

\def\signmarc{\bigskip \begin{center} {\sc
Marc Briant\par\vspace{3mm}
Universit\'e de Paris,\par 
MAP5, CNRS UMR 8145 \par 
F-75006 Paris, France \par
\vspace{3mm}
e-mail:} \tt{briant.maths@gmail.com} \end{center}}

\def\signarnaud{\bigskip \begin{center} {\sc
Arnaud Debussche\par\vspace{3mm}
ENS Rennes,\par
Avenue Robert Schumann,\par
35170 Bruz, France\par
\vspace{3mm}
e-mail:} \tt{arnaud.debussche@ens-rennes.fr} \end{center}}

\def\signjulien{\bigskip \begin{center} {\sc
Julien Vovelle\par\vspace{3mm}
UMPA UMR 5669 CNRS,\par 
ENS de Lyon site Monod,\par
46, allée d'Italie,\par 
9364 Lyon Cedex 07, France \par 
\vspace{3mm}
e-mail:} \tt{julien.vovelle@ens-lyon.fr} \end{center}}

\begin{document} 

\title[Boltzmann equation with external forces]{The Boltzmann equation with an external force on the torus: Incompressible Navier-Stokes-Fourier hydrodynamical limit}
\author{Marc Briant, Arnaud Debussche, Julien Vovelle}

\begin{abstract}
We study the Boltzmann equation with external forces, not necessarily deriving from a potential, in the incompressible Navier-Stokes perturbative regime. On the torus, we establish Cauchy theories that are independent of the Knudsen number in Sobolev spaces. The existence is proved around a time-dependent Maxwellian that behaves like the global equilibrium both as time grows and as the Knudsen number decreases. We combine hypocoercive properties of linearized Boltzmann operators with linearization around a time-dependent Maxwellian that catches the fluctuations of the characteristics trajectories due to the presence of the force. This uniform theory is sufficiently robust to derive the incompressible Navier-Stokes-Fourier system with an external force from the Boltzmann equation. Neither smallness, nor time-decaying assumption is required for the external force, nor a gradient form, and we deal with general hard potential and cut-off Boltzmann kernels. 
As a by-product the latest general theories for unit Knudsen number when the force is sufficiently small and decays in time are recovered.
\end{abstract}

\maketitle

\vspace*{10mm}



\tableofcontents


\section{Introduction}\label{sec:intro}

The Boltzmann equation is used to model rarefied gas dynamics when particles undergo elastic binary collisions, when one studies the gas from a mesoscopic point of view. It describes the time evolution of $f=f(t,x,v)$: the distribution of the particles constituing the gas in position $x$ and velocity $v$. The equation can be derived from Newton's law under the assumption of rarefied gases \cite{CIP} and it reads
\begin{equation}\label{eq:BE}
\tau\partial_t F + v\cdot\nabla_x F = \frac{1}{\mathrm{Kn}}Q(F,F).
\end{equation}
The parameter $\mathrm{Kn}$ is a physical parameter, called the Knudsen number, that gauges the continuity of the gas. Physically speaking, a small Knudsen number indicates that fluid equations are more accurate to describe the gas. The parameter $\tau$ in \eqref{eq:BE} is a relaxation time.
\par For given ranges of the parameters $\tau$ and $\mathrm{Kn}$, one can show that the physical observables - mass, momentum and energy - of the solution $F$ converge are well approached by solutions of acoustic equations or Euler equation or incompressible Navier-Stokes equations, among others. We refer to \cite{Villani02,StRay,Gol14} for a deep discussions on the matter. We will consider the regime
\[
\tau=\mathrm{Kn}=\eps,
\]
with $\eps\to 0$. Describing by the decomposition $F = \mu + \eps f$ the fluctuations of amplitude $\eps$ of the solution $F$ around a global equilibrium $\mu$, we expect an asymptotic description of $f$ in terms of the incompressible Navier-Stokes-Fourier equations:
\begin{align}
\partial_t u - \nu \Delta u + u\cdot \nabla u + \nabla p = 0, \nonumber
\\ \nabla \cdot u = 0, \label{eq:NS}
\\ \partial_t \theta - \kappa \Delta \theta + u\cdot \nabla \theta = 0, \nonumber
\end{align}
together with the Boussinesq relation
\begin{equation}\label{eq:Boussinesq}
\nabla(\rho + \theta) = 0.
\end{equation}
It is interesting to mention that due to initial conservation laws for the Boltzmann equation, the Boussinesq equation actually imposes $\rho+\theta =0$, which in turns gives $\eqref{eq:NS}$ \cite{Gol14}.
\par The resulting perturbative Boltzmann equation is
\begin{equation}\label{eq:perturbed BE}
\partial_t f  + \frac{1}{\eps}v\cdot \nabla_x f = \frac{1}{\eps^2}L[f] + \frac{1}{\eps}\tilde{Q}(f,f)
\end{equation}
where $L$ is a linear operator. 
\par Describing the evolution of the macroscopic parameters, the density, the momentum and the energy associated to $f$ as $\eps$ tends to $0$ has been the subject of numerous works starting from the \textit{a priori} very weak convergence given by the Bardos-Golse-Levermore program \cite{BardosGolseLevermore91} and using a wide range of tools from spectral theory in Fourier space \cite{EllPin,BarUka} to the setting of renormalized solutions \cite{GolSt1,GolSt2}. We point out \cite{Gu4,Bri3} in particular as they rely on two different manifestations of a very important property of the Boltzmann linear operator $L$: its hypocoercivity, which will play a central part in our study. Note that one may differentiate here the perturbative approach of References \cite{BarUka,Gu4,Bri3}, for example, from the approach ``in the large'' of \cite{BardosGolseLevermore89,BardosGolseLevermore91,BardosGolseLevermore93,BardosGolseLevermore98,BardosGolseLevermore00,
LionsMasmoudi01,GolseLevermore02,GolSt1,GolSt2,LevermoreMasmoudi10,Arsenio12}.

\bigskip
The present article focuses on the Boltzmann equation when the gas under consideration is evolving on the $d$-dimensional torus $\T^d$ and is influenced by an external force $\vec{E_t}(x)$. We would like to derive the incompressible Navier-Stokes-Fourier hydrodynamical limit of the latter. In this setting, the Boltzmann equation reads, for $(t,x,v)$ in $[0,T_{\mbox{\footnotesize{max}}})\times\T^d\times\R^d$,
\begin{equation}\label{eq:BE force}
\partial_t F(t,x,v) + \frac{1}{\eps} v\cdot\nabla_x F(t,x,v) + \eps \vec{E_t}(x) \cdot \nabla_vF(t,x,v) = \frac{1}{\eps} Q(F,F)(t,x,v).
\end{equation}
The bilinear operator $Q(g,h)$ is given under its symmetric form:
$$Q(g,h) =  \frac{1}{2}\int_{\R^d\times \mathbb{S}^{d-1}}\Phi\left(|v - v_*|\right)b\left( \mbox{cos}\theta\right)\left[h'g'_*+h'_*g' - hg_*-h_*g\right]dv_*d\sigma,$$
where $f'$, $f_*$, $f'_*$ and $f$ are the values taken by $f$ at $v'$, $v_*$, $v'_*$ and $v$ respectively. Define:
$$\left\{ \begin{array}{rl}&\displaystyle{v' = \frac{v+v_*}{2} +  \frac{|v-v_*|}{2}\sigma} \vspace{2mm} \\ \vspace{2mm} &\displaystyle{v' _*= \frac{v+v_*}{2}  -  \frac{|v-v_*|}{2}\sigma} \end{array}\right., \: \mbox{and} \quad \mbox{cos}\:\theta = \left\langle \frac{v-v_*}{\abs{v-v_*}},\sigma\right\rangle .$$
All along this paper we consider the Boltzmann equation with assumptions
\begin{itemize}
\item[(H1)] \textit{Hard potential} or \textit{Maxwellian potential} ($\gamma=0$), that is to say there is a constant $C_\Phi >0$ such that
\begin{equation}\label{HypH1}
\Phi(z) = C_\Phi z^\gamma \:,\:\: \gamma \in [0,1].
\end{equation}
\item[(H2)] Strong Grad's \textit{angular cutoff} \cite{Grad58}, expressed here by the fact that we assume the non-negative function $b$ to be $C^1$ with the following controls
\begin{equation}\label{HypH2}
\forall z \in [-1,1], \: b(z), \abs{b'(z)} \leq C_b.
\end{equation}
\end{itemize}

\begin{remark} We may relax \eqref{HypH1} into $\Phi(z) \asymp z^\gamma$, in the sense that $C^1_\Phi z^\gamma\leq\Phi(z)\leq C^2_\Phi z^\gamma$ for all $z$, where $C^1_\Phi$ and $C^2_\Phi$ are two positive constants. We may also assume, instead of \eqref{HypH2}, that
\begin{equation}\label{GradCutOff}
\sup_{z\in(-1,1)} b(z)<+\infty,
\end{equation}
and that the non-degeneracy hypothesis
\begin{equation}\label{NDb}
\inf_{\sigma_2,\sigma_3\in\mathbb{S}^{d-1}}\int_{\sigma_3\in\mathbb{S}^{d-1}}\min\{b(\sigma_1\cdot\sigma_3),b(\sigma_2\cdot\sigma_3)\}d\sigma_3>0,
\end{equation} 
is satisfied. Under \eqref{HypH2}, we can use \cite{BarMou} to get a spectral gap estimate on the linearized operator $L$ (see \eqref{spectral gap L}), while, under \eqref{GradCutOff}-\eqref{NDb}, this is \cite{Mou} that can be applied.
\end{remark}


There are two direct observations one can make comparing $\eqref{eq:BE force}$ to the standard Boltzmann equation $\eqref{eq:BE}$. Firstly, the conservation of momentum and energy do not hold, and we are only left with the \textit{a priori} mass conservation
\begin{equation}\label{eq:conservation mass}
\forall t \in [0,T_{\mbox{\footnotesize{max}}}), \quad \frac{d}{dt}\int_{\T^d\times\R^d} F(t,x,v)dxdv =0.
\end{equation}
Secondly, the global equilibrium of the Boltzmann equation
\begin{equation}\label{eq:mu}
\forall v \in R^d,\quad \mu(v) = \frac{1}{(2\pi)^{d/2}}e^{-\frac{\abs{v}^2}{2}},
\end{equation}
which satisfies $Q(\mu,\mu) =0$ is no longer a stationary solution to $\eqref{eq:BE force}$. However, as $\eps$ vanishes we expect the dynamics of the Boltzmann equation with external force to converge towards $\mu(v)$. We aim at constructing an existence and uniqueness theory in Sobolev spaces for solutions to $\eqref{eq:BE force}$ uniformly in $\eps$. We shall look for solutions in a perturbative setting, mimicking the classical decomposition $F= \mu +\eps f$, that will catch the hydrodynamical regime of the incompressible Navier-Stokes-Fourier with external force. More precisely, we intend to show that if, at initial time, $F_0$ is sufficiently close to $\mu(v)$, then so is $F(t)$. Moreover the perturbations of the mass, momentum and energy:
\begin{align}
\rho_\eps(t,x) &= \eps^{-1}\int_{\R^d} \cro{F(t,x,v) - \mu(v)}dv \label{eq:mass}
\\u_\eps(t,x) &= \eps^{-1}\int_{\R^d} v\cro{F(t,x,v) - \mu(v)}dv \label{eq:momentum}
\\\theta_\eps(t,x) &= \eps^{-1}\int_{\R^d} \frac{\abs{v}^2-d}{\sqrt{2d}}\cro{F(t,x,v) - \mu(v)}dv \label{eq:temperature}
\end{align}
converge to $(\rho,u,\theta)$, which are Leray solutions to the following Navier-Stokes-Fou\-rier's system \eqref{eq:NS force}-\eqref{eq:Fourier force}: 
\begin{align}
\partial_t u - \nu \Delta u + u\cdot \nabla u + \nabla p = \frac{\vec{E_t}(x)}{2}, \nonumber
\\ \nabla \cdot u = 0, \label{eq:NS force}
\\ \partial_t \theta - \kappa \Delta \theta + u\cdot \nabla \theta = 0, \label{eq:Fourier force}
\end{align}
together with the Boussinesq relation $\eqref{eq:Boussinesq}$. Leray solutions of the latter means weak solutions integrated against test functions with null divergence. 
We show that $(\rho,u,\theta)$ are solutions in this weak sense, but, due to the estimates in Sobolev spaces with high indexes that we obtain (\textit{cf.} Theorem~\ref{theo:hydro lim}), these solutions are classical, regular solutions close to the equilibrium state $(1,0,1)$..


\bigskip
The present hydrodynamical problem has not been addressed yet in the mathematical litterature, and even the works with $\eps=1$ on the Boltzmann equation with an external force $\eqref{eq:BE force}$ in full generality are scarce. The main issue being that ve\-lo\-ci\-ty derivatives can grow very rapidly. To our knowledge only \cite{DUYZ2008} deals with general $\vec{E_t}(x)$: they solve the perturbative Cauchy theory around $\mu$ in Sobolev spaces for $\eps=1$ as long as the force $\vec{E_t}(x)$ is small and decreases to $0$ as time increases (the latter assumption is removable if $d\geq 5$ or if one solely deals with linear terms). 
\par There have been several studies for $\eps=1$ when the force comes from a potential $\vec{E_t}(x) = \nabla_x V_t(x)$. The latest result in this setting seems to be \cite{Kim2014} and deals with large potential in an $L^\infty$ framework, we refer to the references therein for the potential force framework. This framework is however irrelevant to derive Navier-Stokes-Fourier system with force because one can only construct Leray solutions from Boltzmann-type equations and such solutions do not see gradient terms.
\par At last we would like to present a related issue, still for $\eps=1$, that is when the force is nonlinear: $\vec{E_t}(x)= \vec{E_t}[f](x)$. This happens in electromagnetism for instance. Several results have been obtained in these settings in Sobolev space for perturbation of the global equilibrium. The advantage of this nonlinearity is a feedback that keeps the smallness of the force along the flow. We point out Vlasov-Poisson-Boltzman equations \cite{Guo2002,DYZ2012,DYZ2013,XXZ2013,XXZ2017} or  Vlasov-Maxwell-Boltzmann equations \cite{DLY2013}, the strategies of which will prove themselves useful in our methods. See also \cite{ArsenioSaintRaymond2019} for a non-perturbative approach to those systems.
\par One of the main issue when dealing with the Boltzmann equation with an external force comes from the fact that the perturbative regime $F = \mu + f$ gives rise to a differential equation in $f$ that includes the term $\vec{E_t}(x)(x)\cdot v f$. The latter generates a loss of weight in standard Sobolev estimates. The latest result we are aware of for general non-potential forces comes from \cite{DUYZ2008}, where the authors work in the whole spatial domain $\R^d$ close to $\mu$ and $\vec{E_t}(x)$ is assumed to be small and time-decaying like $(1+t)^{-\alpha}$ if $3\leq d < 5$. Moreover their collision operator must satisfy the hard spheres assumption $\gamma=1$ and $b(\cos\theta) =1$. Working with a hard sphere kernel was mandatory in \cite{DUYZ2008} to compensate the loss of weight in $v$, by means of the negative feedback of the linear Boltzmann operator, which generates a gain of $1+\abs{v}^\gamma$ (see \eqref{spectral gap L}). Note however, that when only studying the semigroup generated by the linear part of the pertubative regime they do not need any time-decay for $\vec{E_t}$. It is important to understand that $\mu$ is no longer a stationary state when $\vec{E_t}(x)\neq 0$ so one hope that $\mu$ shall be stable when the force is very small or in the limit $\eps$ tends to $0$ where formal Chapman-Enskog expansion easily shows that the first order term must be $\mu$. To deal with non small force we propose a different regime.
\par The idea we have arises from the time-dependent norms proposed in \cite{DYZ2013,XXZ2017,DLY2013}, that  compensate the increase of weight due to the nonlinearity of the external force. On the other hand, in a completely different setting, \cite{AcevesSanchezCesbron2019} linked the external force in a fractional Vlasov-Fokker-Planck equation to a new equilibrium that evolves with the external force. The new equilibrium can be explicitely written for the fractional Vlasov-Fokker-Planck equation whereas the non-local part of the Boltzmann equation seems to prevent such a direct treatment. However, we try to combine the two point of views described above : we cannot explicitely extract a new equilibrium for the Boltzmann equation with external force so we fake it by studying the equation around a Maxwellian distribution that depends on $\vec{E_t}(x)$. Such an approach sees the external force as a fluctuation of the classical characteristics of the Boltzmann equation rather than a direct interaction on the solution. We are therefore able to relax the hard sphere assumption, as the loss of weight generated by the external force is effectively compensated, although not by the non-positivity of the linear operator, but thanks to the negative feedback offered by the fluctuation of the Maxwellian. When $\vec{E_t}(x)$ is time-decaying, not only our strategy works for general hard potential kernels with angular cut-off, but it also enables to treat large forces. The core of the proof relies on the construction of twisted Sobolev norms, in the spirit of \cite{MouNeu,Bri3}, which pushes out the hypocoercivity of the Boltzmann linear equation. Namely, the commutator $[v\cdot \nabla_x,\nabla_v] = -\nabla_x$ offers a full negative feedback on $x$ derivatives and one thus would like to work with functional of the form
$$\norm{f}^2= a\norm{f}^2_{L^2_{x,v}}+ b\norm{\nabla_xf}^2_{L^2_{x,v}}+c\norm{\nabla_vf}^2_{L^2_{x,v}} + d\langle \nabla_x f, \nabla_v f\rangle_{L^2_{x,v}}$$
and equivalently in $H^s_{x,v}$ regularity.The mixed term in the twisted norm uses the commutator property and is sufficient in the classical case $\vec{E_t}=0$. The presence of the external force, however, requires a more subtle use of this mixed part that will have to compensate much more terms arising from pure spatial derivatives. We shall see the interplay between the negative feedback offered by the commutator on one side and the one offered by the fluctuation on the other side. The main issue being the negative feedback coming only from the orthogonal part of the solution when dealing with pure $x$ derivatives. Using commutator for fixed pure $x$-derivatives proved itself sufficient for the classical Boltzmann equation $\vec{E_t}=0$ but in our case they have to be dealt with at the same time.
\par Unfortunately, when $\vec{E_t}(x)$ does not display a time-decaying property we can only use this strategy on fixed time intervals $[0,T_0]$, for any $T_0 >0$, but not globally in time. However, and of important note, in the hydrodynamical regime $\eps \to 0$ our method provides solutions close to the global maxwellian $\mu(v)$.

All these thoughts make us look at the perturbative regime around
\begin{equation}\label{eq:M}
\forall (t,v) \in [0,+\infty)\times\R^d, \quad M(t,v) = \frac{e^{-\eps^{1+\eu}\frac{A}{1+t}}}{(2\pi)^{d/2}}e^{-\frac{\abs{v}^2}{2}\pa{1+\eps^{1+\eu}\frac{a}{1+t}}} = \mu e^{-\eps^{1+\eu}\frac{A+a\frac{\abs{v}^2}{2}}{1+t}}
\end{equation}
where $A$ and $a$ stand for positive constant that we shall define in due time. Of core importance, $\eu$ has to belong to $(0,1)$. We refer to Remark \ref{rem:e<1} and Remark \ref{rem:e>0} to understand that when $\eu=0$ then the fluctuation $M$ is not close enough to $\mu$ to perform a relevant hydrodynamical limit whereas when $\eu=1$ the fluctuation goes to fast towards $\mu$ compared to the variations of the characteristics.
\par We study the perturbative regime
\begin{equation}\label{eq:perturbation}
\forall (t,x,v)\in \R^+\times \T^d \times \R^d, \quad F(t,x,v) = M(t,x) + \eps M^{\frac{1}{2}}f(t,x,v)
\end{equation}
which leads to the following perturbative equation
\begin{equation}\label{eq:perturbative BE}
\partial_t f + \frac{1}{\eps}v\cdot\nabla_x f + \eps \vec{E_t}(x)\cdot\nabla_vf + \eps \mathcal{E}(t,x,v) f = \frac{1}{\eps^2}L[f] + \frac{1}{\eps}\Gamma[f,f] -2\mathcal{E}(t,x,v)M^{1/2}
\end{equation}
where $L$ and $\Gamma$ are respectively the standard linear and bilinear perturbative Boltzmann operators around $M$
\begin{eqnarray*}
L[f] &=& \frac{2}{\sqrt{M}}Q(M,\sqrt{M}f)
\\ \Gamma[f,f] &=& \frac{1}{\sqrt{M}}Q(\sqrt{M} f,\sqrt{M}f)
\end{eqnarray*}
and we defined the perturbative force term
\begin{equation}\label{eq:E ronde}
\forall (t,x,v)\in \R^+\times \T^d \times \R^d, \quad \mathcal{E}(t,x,v) = \frac{1}{2}\pa{\eps^\eu\frac{A+a\frac{\abs{v}^2}{2}}{(1+t)^{2}} - \pa{1+\eps^{1+\eu}\frac{a}{1+t}} \vec{E_t}(x)\cdot v}. 
\end{equation}

We conjecture that, by use of the maxwellian regularising properties of the compact part of $L$, one could directly solve the Cauchy problem around a global maxwel\-lian $F = \mu +\eps \sqrt{\mu} f$ for $f$ in $H^s_{x,v}\pa{e^{\eps^2\frac{A+a\abs{v}^{1+0}}{1+t}}}$ when $\vec{E_t}$ or $\eps$ are sufficiently small. However, it would implies some technicalities we did not want to tackle in the present manuscript, where we are only interested in the limit when $\eps$ vanishes: working in $L^2_{x,v}$ framework makes usual properties of $L[f]$ directly applicable, thus our proofs only emphasizes the characteristics fluctuations. Moreover, the strategy we use enables to compensate a quadratic loss of weight $\abs{v}^2$ (rather than the sole $\abs{v}$ specific to the present problem), which may suit further investigations for more complex forces.

\section{Main results}

\subsection{Notations}

For $j=(j_1,\dots,j_d)$ and $l=(l_1,\dots,l_d)$ multi-indexes we define
$$\partial^j_l f = \frac{\partial^{\abs{l}}}{\partial x_1^{l_1}\cdots \partial x_d^{l_d}}\frac{\partial^{\abs{j}}}{\partial v_1^{j_1}\cdots \partial v_d^{j_d}}f.$$
And we define the multi-index: $\forall i \in \br{1,\dots,d},\: \delta_i=(\delta_{ik})_{1\leq k\leq d}$.
For clarity purposes, and as it plays a central role for the linearized Boltzmann operator we shall use the shorthand notation
$$
\forall \beta \in \R,\quad \norm{f}_{L^2_\beta}=  \pa{\int_{\T^d\times\R^d} f(t)^2\langle v \rangle^\beta \:dxdv}^{1/2},\quad \langle v \rangle:=(1+|v|^2)^{1/2}.
$$
Finally, we shall index by $x$, $v$ or $(x,v)$ the norms that will be used::
$$\norm{f}_{L^2_x} = \pa{\int_{\T^d}f(x,v)^2dx}^{\frac{1}{2}},\quad \norm{f}_{L^2_v} = \pa{\int_{\R^d}f(x,v)^2dv}^{\frac{1}{2}}, $$
$$\norm{f}_{L^2_{x,v}} = \pa{\int_{\T^d\times\R^d}f(x,v)^2dxdv}^{\frac{1}{2}}.$$
The same notations apply for Sobolev spaces $H^s_x$ (only $x$-derivatives), $H^s_v$ (only $v$-derivatives) and $H^s_{x,v}$ (both derivatives).
\par In what follows any positive constant depending on a parameter $\alpha$ will be denoted $C_\alpha$. Note that we will not keep track on the dependencies over $d$, $\gamma$ or $b(\cos \theta)$.


\subsection{Results on Cauchy theories}

When $\vec{E_t}$ is a given force we shall prove the following Cauchy problem. We recall Definition $\eqref{eq:M}$ of the fluctuation of a global Maxwellian
$$\forall (t,v) \in [0,+\infty)\times\R^d, \quad M(t,v) = \frac{e^{-\eps^{1+\eu}\frac{A}{1+t}}}{(2\pi)^{d/2}}e^{-\frac{\abs{v}^2}{2}\pa{1+\eps^{1+\eu}\frac{a}{1+t}}} \quad\mbox{with}\quad e\in(0,1).$$
We get a Cauchy theory under the perturbative regime around a given $M$.

\begin{theorem}\label{theo:Cauchy deterministic}
Let the Boltzmann operator satisfies hypotheses $(H1)-(H2)$ and let $s$ be in $\N$. Further assume that $\vec{E_t}$ verifies 
\begin{equation}\label{HypE}
\norm{\vec{E_t}(x)}_{L^\infty_tW^{s,\infty}_x} \leq C_E,\qquad\forall t\geq 0,\int_{\T^d}{\vec{E_t}(x)}dx=0.
\end{equation}
There exists $s_0\in \N^*$ such that for any $s\geq s_0$ the following holds. Let $T_0 >0$ and $C_\mathrm{in}\geq 0$. There exists $\eps_{T_0,C_\mathrm{in},E,s} >0$ such that if $\eps=1$ or $0<\eps <\eps_{T_0,C_\mathrm{in},E,s}$, there exists a norm 
$$
\norm{\cdot}_{\mathcal{H}^s_{x,v}} \sim \sum\limits_{\abs{l}\leq s}\norm{\partial_l^0\cdot}_{L^2_{x,v}} + \eps \sum\limits_{\underset{\abs{j}\geq 1}{\abs{j}+ \abs{l}\leq s}}\norm{\partial_l^j\cdot}_{L^2_{x,v}}
$$
and $A_{T_0,E,s}$, $a_{T_0,E,s}$, $\delta_{T_0,E,s}$,  $C_{T_0,E,s}>0$ such that if $F_{in} = M + \eps \sqrt{M}\:f_{in}$ with
$$
\norm{f_{in}}_{\mathcal{H}^s_{x,v}} \leq \delta_{T_0,E,s} \quad\mbox{and}\quad \abs{\int_{\T^d\times\R^d}\pa{\begin{array}{c} 1 \\ v \\ \abs{v}^2\end{array}}f_{in}(x,v)\sqrt{M}|_{t=0}\:dxdv} \leq C_\mathrm{in} \eps^{\eu},
$$
then there exists a unique solution $F = M + \eps \sqrt{M}\:f$ on $[0,T_0)$ to the Boltzmann equation with external force $\eqref{eq:BE force}$ and it satisfies
$$\forall t \in [0,T_0),\quad \norm{f(t)}_{\mathcal{H}^s_{x,v}} \leq \max\br{\norm{f_{in}}_{\mathcal{H}^s_{x,v}}, C_{T_0,E,s}}.$$
All the constants can be computed explicitly and are independent of $\eps$.
\end{theorem}

Let us make a few comments about the theorem above.

\begin{remark}\label{rem:Cauchy deterministic}
\begin{itemize}
\item The $\mathcal{H}^s_{x,v}$-norm is defined by $\eqref{eq:Hsxv}$ and the equivalence of norm is independent of $\eps$.
\item The hypothesis of cancellation of the first Fourier coefficient $\widehat{\vec{E_t}}(0)$ in $\eqref{HypE}$ is not really necessary\footnote{However, as such, it ensures a condition of quasiconservation
of the total momentum in Navier-Stokes equation, which is used to obtain the estimate $\eqref{Moment1OK}$}. Since the Boltzmann operator is commuting with translations in $v$, the change of variable
\[
v'=v+\eps w_t,\quad w_t:=\int_0^t\widehat{\vec{E_s}}(0)ds
\]
operates a reduction to the case where $\widehat{\vec{E_t}}(0)=0$. In the case of a general forcing term satisfying only the bound by $C_E$ in \eqref{HypE}, the result of Theorem~\ref{theo:Cauchy deterministic} holds true, except that the decomposition of $F$ has to be modified into the following expansion:
\[
F(t,x,v)= M(t,v-\eps w_t) + \eps \sqrt{M}(v-\eps w_t)\:f(t,x,v-\eps w_t),
\]
where $f$ is solution to \eqref{eq:perturbative BE} with a forcing term $\vec{E_t}^\prime:=\vec{E_t}-\widehat{\vec{E_t}}(0)$.

\item We get a local existence result for $\eps=1$ for a non-small, non time-decreasing force and with more general kernels than considered formerly (only hard spheres has been obtained to our knowledge \cite{DUYZ2008}). This is the first result of this kind we are aware of. Moreover, we get close-to-global Maxwellian $\mu$ existence for small $C_E$, recovering and extending on the torus the latest results (see Remark \ref{rem:small eps}).
\item We agree that in the case $\eps=1$ when $C_E$ or $T_0$ are taken larger and larger, the fluctuation $M$ is getting closer to $0$ and our problem thus boils down to the perturbative study around a vacuum state. However, as proven in Corollary \ref{cor:Cauchy deterministic}, in the regime of small epsilon we obtain a perturbative theory around the classical Maxwellian $\mu$.
\item We do not have to impose any decay in time on the force. On the contrary a polynomial decay is used, for instance, in \cite{DUYZ2008}. As explained in the introduction, this comes from the use of a time-dependent Maxwellian as reference state. We think that our strategy is applicable if one looks at solution $F = \mu + \eps \sqrt{\mu} h$ with $h$ belonging to $H^s_{x,v}\pa{e^{\eps^{1+\eu}\frac{A+a\abs{v}^{1+0}}{1+t}}}$ for small $\eps$, by use of a gain of integrability of $K= L + \nu$. It would have been more technical to treat this case, and we thus decided to take a clearer approach which is sufficient to deal with the issue of the hydrodynamical limit.
\item Note that the hypothesis
\begin{equation}\label{HypSmallGlobalMoments}
\abs{\int_{\T^d\times\R^d}\pa{\begin{array}{c} 1 \\ v \\ \abs{v}^2\end{array}}f_{in}(x,v)\sqrt{M}|_{t=0}\:dxdv} \leq C_\mathrm{in} \eps^{\eu}
\end{equation}
says that some global moments of $f_\mathrm{in}$ in both the space and velocity variables are small with $\eps$. Actually, we may replace the right-hand side $C_\mathrm{in} \eps^{\eu}$ of \eqref{HypSmallGlobalMoments} by any quantity that tends to $0$ with $\eps$. In the following Remark~\ref{rk:WellPreparedInitialData}, we comment the incidence of \eqref{HypSmallGlobalMoments} on the initial data for the solution to $\eqref{eq:NS force}$-$\eqref{eq:Boussinesq}$.

\end{itemize}
\end{remark}

As a corollary of Theorem \ref{theo:Cauchy deterministic} we obtain a global perturbative Cauchy theory close to $\mu$.

\begin{cor}\label{cor:Cauchy deterministic}
Under the assumptions of Theorem \ref{theo:Cauchy deterministic} there exists $\eps_{T_0,E,s}$, $\delta_{T_0,E,s}$ and $C_{T_0,E,s}$ such that if $0 < \eps \leq \eps_{T_0,E,s}$ and $F_\mathrm{in} = \mu + \eps \sqrt{\mu}\:f_\mathrm{in}$ with
$$\norm{f_\mathrm{in}}_{\mathcal{H}^s_{x,v}} \leq \delta_{T_0,E,s}\quad\mbox{and}\quad \int_{\T^d\times\R^d}\pa{\begin{array}{c} 1 \\ v \\ \abs{v}^2\end{array}}f_\mathrm{in}(x,v)\sqrt{\mu}\:dxdv = 0,$$
then there exists a unique solution $F = \mu + \eps \sqrt{\mu}\:f$ on $\R^+$ to the Boltzmann equation with external force $\eqref{eq:BE force}$ and it satisfies
$$\forall t \in [0,T_0),\quad \norm{f(t)}_{\mathcal{H}^s_{x,v}} \leq \max\br{\norm{f_\mathrm{in}}_{\mathcal{H}^s_{x,v}}+\eps^\eu C_{T_0,E,s}, C_{T_0,E,s}}.$$
Again, all the constants could be computed explicitly and are independent of $\eps$.
\end{cor}
\begin{remark}\label{rem:small eps} Two remarks are important at this point.
\begin{itemize}
\item We emphasize that taking $\eps$ sufficiently small could be seen as requiring $C_E$ to be sufficiently small in our estimates. The strenght of our result is that the resulting Incompressible Navier-Stokes limit can display non small force $\vec{E_t}$.
\item We point out that $T_0=+\infty$ is not reached in our study because of the negative return of our fluctuation that only works for finite $T_0$ (see Proposition \ref{prop:estimates E ronde}). However, adapting our proofs when $\norm{\vec{E_t}} \leq \frac{C_E}{(1+t)^\alpha}$ with $\alpha >1$ around $M = \mu e^{\eps^{1+\eu}\frac{A+a\abs{v}^2}{(1+t)^{\alpha-1}}}$ gives that not only Proposition \ref{prop:estimates E ronde} is true for $T_0=+\infty$ but also Corollary \ref{cor:Cauchy deterministic} holds globally in time and yields a polynomial time decay (see Remark \ref{rem:time decay}): these are the results of \cite{DUYZ2008} when $\eps=1$ which we recover on the torus.
\end{itemize}
\end{remark}


\subsection{Result on the hydrodynamical limit}

Corollary \ref{cor:Cauchy deterministic} states that if one relabels the solution $f_\eps = \eps^{-1}(F-\mu)\mu^{-\frac{1}{2}}$ then we have uniform bounds on $f_\eps$ in Sobolev spaces and an existence theorem that does not depend on $\eps$ for the initial data. This yields the following convergence result.

\begin{theorem}\label{theo:hydro lim}
Let $T_0 >0$ and $F_\eps$ be the solution built in Corollary \ref{cor:Cauchy deterministic} on $[0,T_0]$ and define $f_\eps = \eps^{-1}\frac{F-\mu}{\sqrt{\mu}}$. Then the sequence $\pa{f_\eps}_{\eps>0}$ converges (up to an extraction) weakly-* in $L^\infty_{[0,T_0]}H^s_{x,v}$ towards an infinitesimal Maxwellian:
\begin{equation}\label{CVtoNS}
\lim\limits_{\eps \to 0} f_\eps(t,x,v) = \cro{\rho(t,x) + v\cdot u(t,x) + \frac{\abs{v}^2-d}{2}\theta(t,x)}\sqrt{\mu},
\end{equation} 
where $(\rho,u,\theta)$ solves the incompressible Navier-Stokes-Fourier system in the sense of Leray with force $\eqref{eq:NS force}$ together with the Boussinesq equation $\eqref{eq:Boussinesq}$.
\end{theorem}

\begin{remark}\label{rk:WellPreparedInitialData} If the data are well prepared in the sense that $f_\mathrm{in}$ is of the form 
\begin{equation}\label{finKerL}
f_\mathrm{in}(x,v)= \cro{\rho_\mathrm{in}(x) + v\cdot u_\mathrm{in}(x) + \frac{\abs{v}^2-d}{2}\theta_\mathrm{in}(x)}\sqrt{\mu},
\end{equation}
with $\nabla\cdot u_\mathrm{in}=0$ and $\rho_\mathrm{in}+\theta_\mathrm{in}=0$ and 
\begin{equation}\label{CsqceHypSmallGlobalMoments}
\int_{\T^d}\rho_\mathrm{in}(x)dx=0,\quad\int_{\T^d}u_\mathrm{in}(x)dx=0
\end{equation}
(\eqref{CsqceHypSmallGlobalMoments} is a consequence of \eqref{HypSmallGlobalMoments}), then we expect (by adaptation of the arguments of the proof of \cite[Theorem~2.5]{Bri3}) to have strong convergence in \eqref{CVtoNS} in the space $C([0,T];H^s_xL^2_v)$, and $(\rho,u,\theta)$ to be a regular solution to $\eqref{eq:NS force}$-$\eqref{eq:Boussinesq}$ on $[0,T_0]$ with initial datum $(\rho_\mathrm{in},u_\mathrm{in},\theta_\mathrm{in})$.
\end{remark}



\section{Properties and estimates on the external operator and the Boltzmann operator}\label{sec:estimates operators}

In the present section we focus on the linear operator $L[f]$, the multiplicative $\mathcal{E}(t,x,v)$ and then on all the operators appearing in the perturbed equation $\eqref{eq:perturbative BE}$.


\subsection{Estimates on the external operator $\mathcal{E}(t,x,v)$}\label{subsec:external operator}

In this section, we will give some estimates on $\mathcal{E}(t,x,v)$ and fix the constants $a$, $A$ and $\alpha$ for the rest of the manuscript. A consequence on the notations used in the paper is that we will drop the specific dependence of constants on the values $a$, $A$ and $\alpha$. Since $\vec{E_t}(x)$ is a datum of the problem on which no smallness assumption is required, we will also drop the possible dependency on $E$: $C_{a,A,\alpha,E}=C$, except in the Proposition below so that the reader can clearly see the improvement we can make when $\vec{E_t}$ decreases in time, that is $C_E \sim (1+t)^{-\alpha}$ as mentionned in Remark \ref{rem:small eps}.

\begin{prop}\label{prop:estimates E ronde} Let $\mathcal{E}$ be defined by \eqref{eq:E ronde} and let $s$ be any integer.
Under the assumption
\[
\norm{\vec{E_t}(x)}_{L^\infty_tW^{s,\infty}_x} \leq C_E<+\infty,
\] 
for any $T_0$ and $\Lambda >0$, there exists $A_{T_0,\Lambda}$ and $a_{T_0,\Lambda} >0$ such that the following properties hold. Positivity:
\begin{equation}\label{eq:positivity E ronde}
\forall (t,x,v) \in [0,T_0) \times \T^d \times \R^d,\quad \mathcal{E}(t,x,v) \geq \eps^\eu\Lambda\pa{1+\abs{v}^2} - \frac{1}{4\eps^\eu}\mathbf{1}_{\eps < 1}.
\end{equation}
Pure spatial derivative estimates: if $s\geq 1$
\begin{equation}\label{eq:spatial derivative E ronde}
\exists C_{T_0,s,\Lambda} >0,\:\forall \abs{l} \geq 1,\:\forall (t,x,v) \in [0,T_0) \times \T^d \times \R^d, \quad \abs{\partial^0_l \mathcal{E}(t,x,v)} \leq C_{T_0,s,\Lambda}\abs{v}.
\end{equation}
Second derivative estimates: if $s\geq 2$
\begin{equation}\label{eq:2nd derivative E ronde}
\exists C_{T_0,s,\Lambda} >0,\:\forall \abs{j} \geq 1 \:\mbox{and}\: \abs{j}+\abs{l} \geq 2,\:\forall (t,x,v) \in [0,T_0) \times \T^d \times \R^d, \quad \abs{\partial^j_l \mathcal{E}} \leq C_{T_0,s,\Lambda}.
\end{equation}
Higher order derivatives in $v$: if $s \geq 3$
\begin{equation}\label{eq:higher order E ronde}
\forall \abs{j}\geq 2\:\mbox{and}\:\abs{j}+\abs{l} \geq 3, \quad \abs{\partial^j_l \mathcal{E}(t,x,v)} = 0.
\end{equation}
\end{prop}

\begin{proof}[Proof of Proposition \ref{prop:estimates E ronde}]
The Proposition is rather straightforward. We recall:
$$\mathcal{E}(t,x,v) = \frac{\eps^\eu(2A+a\abs{v}^2)}{4(1+t)^2} - \frac{1}{2}\pa{1+\eps^{1+\eu}\frac{a}{1+t}} \vec{E_t}(x)\cdot v$$
First a mere Cauchy-Schwarz inequality followed by Young inequality raises that for all $t$ in $[0,T_0)$
\begin{equation*}
\begin{split}
\mathcal{E}(t,x,v) &\geq \eps^\eu\frac{2A+a\abs{v}^2}{4(1+t)^{2}} - \frac{C_E}{2}\abs{v} - \frac{\eps^{1+\eu} a C_E}{2(1+t)}\abs{v}
\\&\geq \eps^\eu \cro{\frac{A}{2(1+t)^2}-\eps^{2}aC_E^2} + \eps^\eu\abs{v}^2\cro{\frac{a}{8(1+t)^2}-\frac{C_E^2}{4}} - \frac{1}{4\eps^\eu}\end{split}
\end{equation*}
So we can first choose $a=a_{T_0,\Lambda}$ sufficiently large so that
$$\frac{a}{8(1+T_0)^2}-\frac{C_E^2}{4} \geq \Lambda$$
and then choose $A=A_{T_0,\Lambda}$ sufficiently large so that
$$\frac{A}{2(1+T_0)^{2}}-aC_E^2\geq \Lambda.$$
This yields the positivity property $\eqref{eq:positivity E ronde}$ because if $\eps=1$ we can make $A$ larger so that $1/4$ is also absorbed. The rest of the estimates are direct computations once constant have been fixed. Note that the constants are independent of $t$ since $\frac{1}{1+t} \leq 1$.
\end{proof}

\begin{remark}
Of important note is the fact that $\Lambda$ shall be fixed later independently of $\eps$: the value of $\Lambda$ is determined in Proposition \ref{prop:sobolev dot deterministic}. We shall carefully keep track of the dependencies in $\Lambda$ to ensure that no bad loop can occur.
\end{remark}


\subsection{Known properties of the Boltzmann operator}\label{subsec: Boltzmann operator}

We gather some well-known properties of the linear Boltzmann operator $L$ (see \cite{Cer, CIP, Vil, GMM} for instance). 
For $1\leq  i\leq d$, let us set
\[
\phi_0(v)=1,\quad\phi_i(v)=v_i,\quad \phi_{d+1}(v)=\frac{1}{2}(|v|^2-d).
\]
The operator $L$ (which is time-dependent) is a closed self-adjoint operator in $L^{2}_v$ with kernel
\begin{equation}\label{eq:ker L}
\mbox{Ker}\left(L\right) = \mbox{Span}\left\{1,v,\abs{v}^2\right\}\sqrt{M} =  \mbox{Span}\left\{\phi_0(v),\dots,\phi_{d+1}(v)\right\}\sqrt{M} ,
\end{equation}
and  $(\phi_0\sqrt{M},\cdot,\phi_{d+1}\sqrt{M})$ is an orthogonal basis of $\mbox{Ker}\left(L\right)$ in $L^2_v$. We denote by $\pi_L$ the orthogonal projection onto $\mbox{Ker}\left(L\right)$ in $L^2_v$:
\begin{equation}\label{piL}
\pi_L(f) = \sum\limits_{i=0}^{d+1} \pa{\int_{\mathbb{R}^d} f(v_*)\bar{\phi}_i(v_*)\sqrt{M(t,v_*)}\:dv_*} \bar{\phi}_i(v)\sqrt{M(t,v)},
\end{equation}
where we have used the normalized family
\[
\bar{\phi}_0=\phi_0,\;\bar{\phi}_1=\phi_1,\dotsc,\;\bar{\phi}_d=\phi_d,\;\;\bar{\phi}_{d+1}=\sqrt{\frac{d}{2}}\phi_{d+1}.
\]
We set $\pi_L^\bot = \mbox{Id} - \pi_L$. The projection $\pi_L(f(x,\cdot))(v)$ of $f(x,v)$ onto the kernel of $L$ is called the fluid part whereas $\pi_L^\bot(f)$ is called the microscopic part.
\par The operator $L$ can be written under the following form
\begin{equation}\label{decomposition L}
L = -\nu(v) + K,
\end{equation}
where $\nu(v)$ is the collision frequency
$$\nu(v) = \int_{\mathbb{R}^d\times\mathbb{S}^{d-1}} b\left(\mbox{cos}\:\theta\right)\abs{v-v_*}^\gamma M_*\:d\sigma dv_*$$
and $K$ is a bounded and compact operator in $L^2_v$. We give a series of estimates on the operators above that have been proved in the case $a=A=0$ in references we gave above. We solely empasize that the constants do not depend on $t$.

\begin{prop}\label{prop:properties L} Let $s$ be in $\N$, $\abs{l}+\abs{j}\leq s$ and $f$ be in $H^s_{x,v}$.
\\The collision frequency is strictly positive
\begin{equation}\label{positivity nu}
\forall v \in \R^d,\quad 0 <\nu_{0,s,\Lambda}\langle v \rangle^{\gamma} \leq \nu(v) \leq \nu_{1,s,\Lambda}\langle v \rangle^{\gamma}.
\end{equation}
The operator $L$ acts on the $v$-variable and has a spectral gap $\lambda_{0,\Lambda} >0$ in $L^2_{x,v}$
\begin{equation}\label{spectral gap L}
\langle \partial_l^0L(f),\partial_l^0f\rangle_{L^2_{x,v}} \leq -\lambda_{0,\Lambda} \norm{\pi_L^\bot(\partial_l^0 f)}_{L^2_\gamma}^2.
\end{equation}
There exists $\lambda_{s,\Lambda},\: C_{s,\Lambda} >0$ such that, if $\abs{j}\geq 1$, then
\begin{equation}\label{estimate Hs L}
\langle \partial^j_l L[f], \partial^j_lf \rangle_{L^2_{x,v}} \leq -\lambda_{s,\Lambda}\norm{\partial^j_l f}^2_{L^2_\gamma} + C_{s,\Lambda}\norm{f}^2_{H^{s-1}_{x,v}}.
\end{equation}
At last we have the following estimates on scalar products: for $0\leq \abs{l}\leq s-1$ and any $\eta_0 >0$ there exists $C_{\Lambda,\eta_0}$ such that:
\begin{equation}\label{estimate scalar product L}
\langle \partial_{l+\delta_i}^0 L[f],\partial_l^{\delta_i}f \rangle_{L^2_{x,v}}
 = \frac{\lambda_{0,\Lambda}C_{\Lambda,\eta_0}}{\eps}\norm{\pi_L^\bot\pa{\partial_{l+\delta_i}^0 f}}_{L^2_\gamma} 
 + \eps \lambda_{0,\Lambda}\eta_0\norm{\partial^{\delta_i}_{l}f}^2_{L^2_\gamma}.
\end{equation}
\end{prop}

Before getting into the proof of Proposition \ref{prop:properties L}, let us emphasize that the multiplication by $\lambda_{0,\Lambda}$ in the scalar product estimate is of core importance for the case $\eps=1$.

\begin{proof}[Proof of Proposition \ref{prop:properties L}]
If we denote by $L_\mu = -\nu_\mu + K_\mu$ the linear operator when $A=a=0$ then the results hold for $L_\mu$, $\nu_\mu$ and $K_\mu$: see for instance \cite{BarMou,Mou} for the spectral gap, \cite[-(H1')+(H2') page 13]{MouNeu} for Sobolev estimates and \cite[Appendix B.2.3 and B.2.5]{Bri3} for the scalar product.
\par The operator $L$ only acts on the velocity variable thus the change of variable  
$$v_*\mapsto \pa{\sqrt{1+\frac{\eps^{1+\eu} a}{1+t}}}^{-1}v_*$$
shows
\begin{equation}\label{link Lmu}
L[f](v) = \pa{1+\frac{\eps^{1+\eu} a}{1+t}}^{-\frac{d+\gamma}{2}}e^{-\frac{\eps^{1+\eu} A}{1+t}}L_\mu[\tilde{f}]\pa{v\sqrt{1+\frac{\eps^{1+\eu} a}{1+t}}}
\end{equation}
where $\tilde{f}(v) = f\pa{\frac{v}{\sqrt{1+\frac{\eps^{1+\eu} a}{1+t}}}}$. As $0\leq (1+t)^{-1} \leq 1$, inequalities $\eqref{positivity nu}-\eqref{spectral gap L}-\eqref{estimate Hs L}$ directly follow from the case $a=A=0$.
\par The scalar product $\eqref{estimate scalar product L}$ is a mere Cauchy-Schwarz inequality with Young inequality with constant $\eta_0$. Let us show that the resulting constant is of the form $C_{\eta_0}\lambda_{0,\Lambda}$. Denoting $h=1+\frac{\eps^{1+\eu} a}{1+t}$, integrating $\eqref{link Lmu}$ against $f$ yields
\begin{equation*}
\begin{split}
\langle L[f],f\rangle_{L^2_{x,v}} &= h^{-\frac{2d+\gamma}{2}}e^{-\frac{\eps^{1+\eu} A}{1+t}}\langle L_\mu[\tilde{f}],\tilde{f}\rangle_{L^2_{x,v}} \leq - h^{-\frac{2d+\gamma}{2}}e^{-\frac{\eps^{1+\eu} A}{1+t}}\norm{\tilde{\pi_{L_\mu}^\bot}\pa{\tilde{f}}}^2_{L^2_\gamma} 
\\&\leq - \lambda_{0,0}h^{-\frac{d+\gamma}{2}}e^{-\frac{\eps^{1+\eu} A}{1+t}}\norm{\pi_{L}^\bot\pa{f}}^2_{L^2_\gamma}.
\end{split}
\end{equation*}
We thus see that $\lambda_{0,\Lambda} = \lambda_{0,0}h^{-\frac{d+\gamma}{2}}e^{-\frac{\eps^{1+\eu} A}{1+t}}$ and $h^{-\frac{d+\gamma}{2}}e^{-\frac{\eps^{1+\eu} A}{1+t}}$ is exactly the quantity appearing in $\eqref{link Lmu}$ so the expected $\eqref{estimate scalar product L}$ follows.
\end{proof}

We conclude the present section with estimates on the bilinear operator $\Gamma[f,g]$.

\begin{prop}\label{prop:estimates Gamma}
Let $s_0 = d+1$ then for any $s\geq s_0$, for any $\abs{j}+\abs{l} \leq s$ and any $\eta_0,\:\eta_0' >0$, the following holds for $f$, $g$ and $h$ in $H^{\abs{j}+\abs{l}}_{x,v}$,
\begin{align*}
\abs{\langle \partial_l^0 \Gamma(g,h), f\rangle_{L^2_{x,v}}} &\leq \frac{\lambda_{0,\Lambda}\eta_0}{\eps}\norm{\pi_L^\bot (f)}^2_{L^2_\gamma}
\\&\quad+ \eps \lambda_{0,\Lambda}C_{s,\eta_0} \pa{\norm{g}^2_{H^s_{x}L^2_v}\norm{h}^2_{H^s_xL^2_{v,\gamma}} + \norm{h}^2_{H^s_{x}L^2_v}\norm{g}^2_{H^s_xL^2_{v,\gamma}}}
\\\abs{\langle \partial_l^j \Gamma(g,h), f\rangle_{L^2_{x,v}}} &\leq \frac{\eta_0'}{\eps}\norm{f}^2_{L^2_\gamma} +  \eps C_{s,\Lambda,\eta_0'} \pa{\norm{g}^2_{H^s_{x,v}}\norm{h}^2_{H^s_\gamma} + \norm{h}^2_{H^s_{x,v}}\norm{g}^2_{H^s_{\gamma}}}.
\end{align*}
\end{prop}
\begin{proof}[Proof of Proposition \ref{prop:estimates Gamma}]
 As above, these estimates have been obtained when $A=a=0$ and therefore the same arguments as before extend them to the general case. We refer the reader to \cite[Appendix A.2]{Bri3} for constructive proofs in the case $A=a=0$. We find the standard control
\begin{align*}
\abs{\langle \partial_l^0 \Gamma(g,h), f\rangle_{L^2_{x,v}}} &\leq C_{s,\Lambda} \pa{\norm{g}_{H^s_{x}L^2_v}\norm{h}_{H^s_xL^2_{v,\gamma}} + \norm{h}_{H^s_{x}L^2_v}\norm{g}_{H^s_xL^2_{v,\gamma}}}\norm{\pi_L^\bot (f)}_{L^2_\gamma},
\\\abs{\langle \partial_l^j \Gamma(g,h), f\rangle_{L^2_{x,v}}} &\leq C_{s,\Lambda} \pa{\norm{g}_{H^s_{x,v}}\norm{h}_{H^s_\gamma} + \norm{h}_{H^s_{x,v}}\norm{g}_{H^s_{\gamma}}}\norm{f}_{L^2_\gamma},
\end{align*}
that we complete with Young inequality. The dependency in $\lambda_{0,\Lambda}$ follows exactly from the same argment as in the proof of Proposition \ref{prop:properties L}.
\end{proof}


\subsection{Estimates for each operator}\label{subsec:estimates deterministic}

The perturbative equation \eqref{eq:perturbative BE} that we shall study can be decomposed as the evolution by $6$ different operators:
\begin{equation*}
\begin{split}
\partial_t f &= - \frac{1}{\eps}v\cdot\nabla_x f - \eps \vec{E_t}(x)\cdot\nabla_vf - \eps \mathcal{E}(t,x,v) f + \frac{1}{\eps^2}L[f] + \frac{1}{\eps}\Gamma[f,f] -2 \mathcal{E}(t,x,v)M^{1/2}
\\&:= \sum\limits_{i=1}^6 S_i(t,x,v).
\end{split}
\end{equation*}
We prove a series of Lemmas to estimate each operator in Sobolev norms. To clarify the computations we shall use the convention that $\partial^j_l f =0$ whenever the multi-indexes $j$ or $l$ contains one negative component. Thus any integration by parts can be computed.

\paragraph{Strategy} Our final aim is to get an estimate on the weighted norm (see \eqref{normweight})
\begin{equation}\label{normweight0}
f\mapsto \sum\limits_{\abs{l}\leq s}\norm{\partial^0_l f}^2_{L^2_{x,v}} + \eps^2 \sum\limits_{\underset{\abs{j}\geq 1}{\abs{l}+\abs{j}\leq s}}\norm{\partial_l^j f}^2_{L^2_{x,v}}.
\end{equation}
Standard energy estimates will provide some gain and loss terms. The gain terms are due
\begin{itemize}
\item to the spectral gap estimates \eqref{spectral gap L} and \eqref{estimate Hs L}: they are 
\begin{equation}\label{GainL}
-\frac{\lambda_{0,\Lambda}}{\eps^2} \sum\limits_{\abs{l}\leq s}\norm{\pi_L^\bot(\partial^0_l f)}^2_{L^2_{x,v}},\quad  -\lambda_{s,\Lambda} \sum\limits_{\underset{\abs{j}\geq 1}{\abs{l}+\abs{j}\leq s}}\norm{\partial_l^j f}^2_{L^2_{x,v}};
\end{equation}
\item to the operator $\mathcal{E}$: associated to the derivative $\partial_l^j$, we have a gain (with weight $|v|^2$) which is $-\frac{\eps^{1+\eu} \Lambda}{2}\norm{\partial_l^jf}^2_{L^2_2}$ (see Lemma~\ref{lem:S3} below).
\end{itemize}
In a procedure that is standard for the derivation of hypocoercive estimates, we also introduce a correction by the twisted terms $\eps\langle \partial_{l+\delta_i}^0 f,\partial_l^{\delta_i}f \rangle_{L^2_{x,v}}$ in \eqref{normweight0}. Note those terms are pondered with a weight $\eps$. Those terms will provide the gain term $-\norm{\partial_{l+\delta_i}^0 f}^2_{L^2_{x,v}}$ (see Lemma~\ref{lem:S2} below). Combining those latter terms with the terms of the second sum in \eqref{GainL}, we obtain (up to the terms of order $0$) a gain of almost a full $H^s_{x,v}$-norm, having no weight $\eps$. This is why the occurrence of a term $C\|f\|_{H^s_{x,v}}^2$ in the forthcoming estimates (Lemma~\ref{lem:S1} to \ref{lem:S6}) is admissible, a control on the size of the constant $C$ being possibly necessary to ensure a good control when all estimates are gathered (which is done step by step in Proposition~\ref{prop:estimate Qli}, Proposition~\ref{prop:estimate djl}, Proposition~\ref{prop:sobolev dot deterministic}). In Proposition~\ref{prop:sobolev dot deterministic}, we also study the evolution of the global moments of $f^\eps$ (in combination with a Poincar\'e-Wirtinger inequality), in order to recover the term of order $0$ that is lacking in our estimates.


\begin{lemma}\label{lem:S1}
Let $s$ be in $\N$ and for $f$ in $H^s_{x,v}$ define $S_1(t,x,v) = - \frac{1}{\eps}v\cdot\nabla_x f$.
\\Then for any $\eta_1 >0$, there exists $C_{\eta_1} >0$ such that for any multi-indexes $l$, $j$ such that $\abs{l}+\abs{j} \leq s$,
$$\abs{\langle \partial_l^j S_1,\partial_l^j f \rangle_{L^2_{x,v}}} \leq \left\{\begin{array}{ll} \frac{\eta_1}{\eps^2}\norm{\partial_l^j f}^2_{L^2_{x,v}} + \frac{1}{\eta_1}\sum\limits_{k=1}^d\norm{\partial^{j-\delta_k}_{l+\delta_k}f}^2_{L^2_{x,v}}\quad & \mbox{if}\quad  \abs{j}\geq 1 \\ 0 & \mbox{if}\quad j=0. \end{array}\right.$$
We have moreover
$$\langle \partial_{l+\delta_i}^0 S_1,\partial_l^{\delta_i}f \rangle_{L^2_{x,v}} = -\frac{1}{2\eps}\norm{\partial_{l+\delta_i}^0 f}^2_{L^2_{x,v}}.$$
\end{lemma}
\begin{proof}[Proof of Lemma \ref{lem:S1}]
By direct computations
\begin{align*}
\langle \partial_l^j\pa{v\cdot \nabla_x f}, \partial_l^j f \rangle_{L^2_{x,v}} &= \sum\limits_{j_1+j_2 = j}\sum\limits_{k=1}^d \int_{\T^d \times \R^d} \pa{\partial^{j_1}_0 v_k} \pa{\partial^{j_2}_{l+\delta_k} f}\partial^j_lf\:dxdv
\\&= \sum\limits_{k=1}^d \int_{\T^d \times \R^d} v_k \pa{\partial^{j}_{l+\delta_k} f}\partial^j_lf + \sum\limits_{k=1}^d \int_{\T^d \times \R^d}  \pa{\partial^{j-\delta_k}_{l+\delta_k} f}\partial^j_lf
\\&= \sum\limits_{k=1}^d \int_{\T^d \times \R^d}  \pa{\partial^{j-\delta_k}_{l+\delta_k} f}\partial^j_lf\:dxdv.
\end{align*}
We used the property $\partial^{j_1}_0(v_k) = 0$ if ($\abs{j_1} \geq 2$) or ($\abs{j_1}=1$ and $j_1\neq \delta_k$). The first result then follows from Cauchy-Schwarz and Young inequalities: for any $\eta_1>0$,
$$\int_{\T^d \times \R^d}  \pa{\partial^{j-\delta_k}_{l+\delta_k} f}\partial^j_lf\:dxdv \leq \frac{\eta_1}{\eps} \int_{\T^d \times \R^d}  \pa{\partial^{j-\delta_k}_{l+\delta_k} f}^2\:dxdv + \frac{\eps}{\eta_1}\int_{\T^d \times \R^d}  \pa{\partial^j_lf}^2\:dxdv.$$
The second equality comes from direct integration by parts.
\end{proof}

\begin{lemma}\label{lem:S2}
Let $s$ be in $\N$ and for $f$ in $H^s_{x,v}$ define $S_2(t,x,v) = - \eps\vec{E_t}(x)\cdot \nabla_v f$.
\\Then for any $\eta_2 >0$, there exists $C_{\eta_2} >0$ such that for any multi-indexes $l$, $j$ satisfying $\abs{l}+\abs{j} \leq s$, we have 
\[
\abs{\langle \partial_l^j S_2,\partial_l^j f \rangle_{L^2_{x,v}}}=0\quad \mbox{if}\quad  l=0,
\]
and
\begin{multline*}
\abs{\langle \partial_l^j S_2,\partial_l^j f \rangle_{L^2_{x,v}}} \\
\leq
\eps^{1+\eu} \eta_2 \Lambda\norm{\partial_l^j f}^2_{L^2_{x,v}} + \eps^{1-\eu}\frac{C_{\eta_2}}{\Lambda }\pa{\sum\limits_{1\leq i,k \leq d}\norm{\partial^{j+\delta_k}_{l-\delta_i}f}^2_{L^2_{x,v}} + \norm{f}^2_{H^{\abs{j}+\abs{l}-1}_{x,v}}},
\end{multline*}
if $|l|>0$, where $\Lambda$ is given by Proposition $\eqref{prop:estimates E ronde}$. We have moreover
$$\abs{\langle \partial_{l+\delta_i}^0 S_2,\partial_l^{\delta_i}f \rangle_{L^2_{x,v}}} \leq \eps^{2+\eu}\eta_2 \Lambda\norm{\partial_{l}^{\delta_i}f}^2_{L^2_{x,v}} + \frac{C_{\eta_2}}{\eps^\eu\Lambda}\sum\limits_{k=1}^d\norm{\partial_{l}^{\delta_k}f}^2_{L^2_{x,v}}+  \frac{C_{\eta_2}}{\eps^\eu}\norm{f}_{H^{\abs{l}}_{x,v}}^2.$$
\end{lemma}
\begin{proof}[Proof of Lemma \ref{lem:S2}]
Here again direct computations give
\begin{align*}
\abs{\langle \partial_l^j \pa{\vec{E_t}(x)\cdot \nabla_v f},\partial_l^j f \rangle_{L^2_{x,v}}} &= \abs{\sum\limits_{k=1}^d \sum\limits_{l_1+l_2 = l}\int_{\T^d\times\R^d}\partial_{l_1}^0E_k(x)\partial^{j+\delta_k}_{l_2}f \partial^j_l f \:dxdv}
\\&= \abs{\sum\limits_{k=1}^d \sum\limits_{\underset{\abs{l_2} < \abs{l}}{l_1+l_2 = l}}\int_{\T^d\times\R^d}\partial^0_{l_1}E_k(x)\partial_{l_2}^{j+\delta_k}f \partial^j_l f \:dxdv}
\\&\leq\norm{\vec{E_t}}_{W^{s,\infty}_x}\sum\limits_{k=1}^d \sum\limits_{\abs{l_2} < \abs{l}}\int_{\T^d\times\R^d}\abs{\partial_{l_2}^{j+\delta_k}f}\abs{\partial^j_l f} \:dxdv
\end{align*} 
and here again combining Cauchy-Schwarz and Young inequality with constant $\eps^\eu\eta_2$ yields the expected result.

Let us look at the second estimate. We have
$$\langle \partial_{l+\delta_i}^0 S_2,\partial_l^{\delta_i}f \rangle_{L^2_{x,v}} =-\eps \sum\limits_{l_1+l_2 = l+\delta_i}\sum\limits_{k=1}^d\int_{\T^d\times\R^d} \partial_{l_1}^0E_k \partial^{\delta_k}_{l_2}f\partial_l^{\delta_i}f \:dxdv.$$
The higher derivative appears when $l_2=l+\delta_i$ and by integration by parts we see
\begin{align*}
\int_{\T^d\times\R^d} E_k \partial^{\delta_k}_{l+\delta_i}f\partial_l^{\delta_i}f\:dxdv  &= -\int_{\T^d\times\R^d} \partial_{\delta_i}^0E_k \partial^{\delta_k}_{l}f\partial_l^{\delta_i}f - \int_{\T^d\times\R^d} E_k \partial^{\delta_k}_{l}f\partial_{l+\delta_i}^{\delta_i}f
\\&= -\int_{\T^d\times\R^d} \partial_{\delta_i}^0E_k \partial^{\delta_k}_{l}f\partial_l^{\delta_i}f - \int_{\T^d\times\R^d} E_k \partial^{\delta_i}_{l}f\partial_{l+\delta_i}^{\delta_k}f
\end{align*}
which implies
$$\int_{\T^d\times\R^d} E_k \partial^{\delta_k}_{l+\delta_i}f\partial_l^{\delta_i}f \:dxdv = -\frac{1}{2}\int_{\T^d\times\R^d}\partial_{\delta_i}^0E_k \partial^{\delta_k}_{l}f\partial_l^{\delta_i}f \:dxdv$$
and therefore the result follows from Cauchy-Schwarz and Young inequalities with constant $\eps^{1+\eu}\eta_2$.
\end{proof}

\begin{lemma}\label{lem:S3}
Let $s$ be in $\N$ and for $f$ in $H^s_{x,v}$ define $S_3(t,x,v) = - \eps\mathcal{E}(t,x,v) f$, where we recall that $\mathcal{E}$ is given by $\eqref{eq:E ronde}$.
\\Then there exists $C_3 >0$ such that for any multi-indexes $l$, $j$ such that $\abs{l}+\abs{j} \leq s$,
$$\langle \partial_l^jS_3,\partial_l^jf \rangle _{L^2_{x,v}} \leq \left\{\begin{array}{l} -\frac{\eps^{1+\eu} \Lambda}{2}\norm{\partial_l^jf}^2_{L^2_2} +\eps^{1-\eu}\frac{\mathbf{1}_{\eps <1}}{4}\norm{\partial_l^jf}^2_{L^2_{x,v}} + \eps^{1-\eu}\frac{C_3}{\Lambda}\norm{f}^2_{H^{\abs{j}+\abs{l}-1}_{x,v}} \\ -\eps^{1+\eu}\Lambda\norm{f}^2_{L^2_2} + \frac{\eps^{1-\eu}}{4}\mathbf{1}_{\eps <1}\norm{f}^2_{L^2_{x,v}} \quad \mbox{if}\quad  \abs{j}+\abs{l}= 0 \end{array}\right.$$
where $\Lambda$ is given by Proposition $\eqref{prop:estimates E ronde}$. Moreover, for any $\eta_3 >0$, there exists $C_{\eta_3} >0$ such that
\begin{equation*}
\begin{split}
\abs{\langle \partial_{l+\delta_i}^0 S_3,\partial_l^{\delta_i}f \rangle_{L^2_{x,v}}} \leq & \eps^{2+\eu}\eta_3 \Lambda\norm{\partial_l^{\delta_i}f}^2_{L^2_2} +\eps^{2-e}\Lambda C_{\eta_3}\norm{\partial_{l+\delta_i}^0f}^2_{L^2_{x,v}} 
\\&+ \eps^{1-\eu}\frac{\mathbf{1}_{\eps <1}}{4}\pa{\norm{\partial_{l+\delta_i}^0f}^2_{L^2_{x,v}}+\norm{\partial_{l}^{\delta_i}f}^2_{L^2_{x,v}}}+ \frac{C_{,\eta_3}}{\Lambda \eps^\eu}\norm{f}^2_{H^{\abs{l}}_{x,v}}.
\end{split}
\end{equation*}
\end{lemma}

\begin{proof}[Proof of Lemma \ref{lem:S3}]
The inequality for $\abs{j}=\abs{l}=0$ is a direct consequence of Proposition \ref{prop:estimates E ronde} and more precisely $\eqref{eq:positivity E ronde}$. When $\abs{j}+\abs{l}>0$  we compute
$$\langle \partial_l^j\pa{\mathcal{E}f},\partial_l^jf \rangle _{L^2_{x,v}} = \sum\limits_{j_1+j_2 = j}\sum\limits_{l_1+l_2 = l} \int_{\T^d\times\R^d} \partial^{j_1}_{l_1} \mathcal{E}\partial^{j_2}_{l_2}f \partial^j_l f \:dxdv.$$
Proposition \ref{prop:estimates E ronde} tells us that most of the derivatives of $\mathcal{E}$ vanish: when ($\abs{j_1}\geq 2$ and $\abs{j_1}+\abs{l_1}\geq 3)$. We therefore decompose the sum into three different parts:
\begin{equation*}
\begin{split}
\sum\limits_{j_1+j_2 = j}\sum\limits_{l_1+l_2 = l} \int_{\T^d\times\R^d} \partial^{j_1}_{l_1} \mathcal{E}\partial^{j_2}_{l_2}f \partial^j_l f \:dxdv =& \int_{\T^d\times\R^d}\mathcal{E}\pa{\partial^{j}_{l}f}^2
\\&+ \sum\limits_{\underset{\abs{l_1}\geq 1}{l_1+l_2 = l}} \int_{\T^d\times\R^d} \partial^{0}_{l_1} \mathcal{E}\partial^{j}_{l_2}f \partial^j_l f
\\&+ \sum\limits_{\underset{\abs{j_1}=1}{j_1+j_2 = j}}\sum\limits_{l_1+l_2 = l} \int_{\T^d\times\R^d} \partial^{j_1}_{l_1} \mathcal{E}\partial^{j_2}_{l_2}f \partial^j_l f 
\\&+ \sum\limits_{\underset{\abs{j_1}=2}{j_1+j_2 = j}} \int_{\T^d\times\R^d} \partial^{j_1}_{0} \mathcal{E}\partial^{j_2}_{l}f \partial^j_l f.
\end{split}
\end{equation*}
Proposition \ref{prop:estimates E ronde} gives us the estimate of the first term, see $\eqref{eq:positivity E ronde}$. In the second and third terms $\abs{\partial^{j_1}_{l_1}\mathcal{E}}$ is bounded by $C_{s}(1+\abs{v})$, see $\eqref{eq:spatial derivative E ronde}- \eqref{eq:2nd derivative E ronde}$. Finally, in the fourth term we have $\abs{\partial^{j_1}_{0} \mathcal{E}}$ bounded by $C_{s}$. Hence using Cauchy-Schwarz and Young inequality with $\eta >0$:
\begin{equation*}
\begin{split}
\langle \partial_l^j\pa{\mathcal{E}f},\partial_l^jf \rangle _{L^2_{x,v}} \geq& \int_{\T^d\times\R^d} \eps^\eu \Lambda\pa{1+\abs{v}^2}\pa{1-2\eta C_{s}}\pa{\partial^{j}_{l}f}^2\:dxdv
\\& - \frac{\mathbf{1}_{\eps <1}}{4\eps^\eu}\norm{\partial_l^jf}^2_{L^2_{x,v}}- \frac{C_{s,\Lambda,\eta}}{\eps^\eu}\norm{f}^2_{H^{\abs{j}+\abs{l}-1}_{x,v}}.
\end{split}
\end{equation*}
We choose $\eta$ sufficiently small and the result follows.

The second estimate is derived in the same spirit. We have
\begin{align*}
\langle \partial_{l+\delta_i}^0 S_2,\partial_l^{\delta_i}f \rangle_{L^2_{x,v}} = \sum\limits_{l_1+l_2 = l+\delta_i}\int_{\T^d\times\R^d} \partial_{l_1}^{0} \mathcal{E} \partial^{0}_{l_2}f\partial_l^{\delta_i}f \:dxdv.
\end{align*}
When $\abs{l_1} \geq 1$ then $\abs{l_2} \leq l$ and $\abs{\partial^0_{l_1}\mathcal{E}} \leq C_{s}\Lambda\abs{v} \leq C_{s}\Lambda\abs{v}$, by Proposition \ref{prop:estimates E ronde}. Therefore using Cauchy-Schwarz and Young inequality with constant $\eps^{1+\eu} \eta_3 >0$ we have
\begin{equation*}
\begin{split}
\abs{\sum\limits_{\underset{\abs{l_1}\geq 1}{l_1+l_2 = l+\delta_i}}\int_{\T^d\times\R^d} \partial_{l_1}^{0} \mathcal{E} \partial^{0}_{l_2}f\partial_l^{\delta_i}f \:dxdv} \leq& \int_{\T^d\times\R^d}\eps^{1+\eu} \eta_3 \Lambda(1+\abs{v}^2)\pa{\partial_l^{\delta_i}f}^2 dxdv 
\\&+ \frac{\Lambda C_{s,\eta_3}}{\eps^{1+\eu}}\norm{f}^2_{H^{\abs{l}}_{x,v}}.
\end{split}
\end{equation*}
At last, when $l_1=0$ we use Proposition \ref{prop:estimates E ronde} to bound $\abs{\mathcal{E}} \leq \eps^\eu C_{s}\Lambda(1+\abs{v}^2)+\frac{\mathbf{1}_{\eps <1}}{4\eps^\eu}$ and get with the Young inequality
\begin{equation*}
\begin{split}
\abs{\int_{\T^d\times\R^d} \mathcal{E} \partial^{0}_{l+\delta_i}f\partial_l^{\delta_i}f \:dxdv} \leq& \eps^{1+\eu}\Lambda \eta_3\int_{\T^d\times\R^d} (1+\abs{v}^2)\pa{\partial_l^{\delta_i}f}^2 dxdv 
\\&+\Lambda C_{s,\eta_3} \int_{\T^d\times\R^d}(1+\abs{v}^2)\pa{\partial_{l+\delta_i}^{0}f}^2 dxdv
\\&+ \frac{\mathbf{1}_{\eps <1}}{4\eps^\eu}\pa{\norm{\partial_{l+\delta_i}^0f}^2_{L^2_{x,v}}+\norm{\partial_{l}^{\delta_i}f}^2_{L^2_{x,v}}}.
\end{split}
\end{equation*}
This concludes the proof.
\end{proof}

It remains to estimate the last operator $S_6$, which is a mere multiplicative operator.

\begin{lemma}\label{lem:S6}
Let $s$ be in $\N$ and for $f$ in $H^s_{x,v}$ define $S_6(t,x,v) = - \mathcal{E}(t,x,v) M^{1/2}$, where we recall that $\mathcal{E}$ is given by $\eqref{eq:E ronde}$.
\\Then for any $\eta_4>0$, there exists $C_{\Lambda,\eta_4} >0$ such that for any multi-indexes $l$, $j$ such that $\abs{l}+\abs{j} \leq s$, we have
$$\langle \partial_l^jS_6,\partial_l^jf \rangle _{L^2_{x,v}} \leq \left\{\begin{array}{ll} \eta_4\norm{\pi_L\pa{\partial_l^{0} f}}^2_{L^2_{x,v}} + C_{\Lambda,\eta_4}\quad & \mbox{if}\quad j=0 \\ \frac{\eta_4}{\eps^2}\norm{\partial_l^{j} f}^2_{L^2_\gamma}+\eps^2C_{\Lambda,\eta_4} \quad &\mbox{if}\quad  \abs{j} \geq 1. \end{array}\right.$$
Moreover, we have
$$\abs{\langle \partial_{l+\delta_i}^0 S_6,\partial_l^{\delta_i}f \rangle_{L^2_{x,v}}} \leq \frac{\eta_4\lambda_{0,\Lambda}}{\eps}\norm{\partial_l^{\delta_i} f}^2_{L^2_\gamma}+ \eps C_{\Lambda,\eta_4}.$$
\end{lemma}
\begin{proof}[Proof of Lemma \ref{lem:S6}]
We can use the estimates on $\mathcal{E}$ derived in Proposition \ref{prop:estimates E ronde}, that we multiply by the Maxwellian $\sqrt{M}$. Looking at the kernel of $L$ given by $\eqref{eq:ker L}$ we see that $\mathcal{E}\sqrt{M}$ belongs to $\mbox{Ker}(L)$ and so does $\partial_l^0\mathcal{E}\sqrt{M}$ for any multi-index $l$. Therefore by Cauchy-Schwarz inequality
\begin{align*}
\abs{\langle \partial_{l}^0 \pa{\mathcal{E}M^{1/2}},\partial_l^{0} f \rangle_{L^2_{x,v}}} &= \abs{\langle \partial_{l}^0 \pa{\mathcal{E}M^{1/2}},\pi_L\pa{\partial_l^{0} f} \rangle_{L^2_{x,v}}}
\\&\leq\norm{\partial_{l}^0 \pa{\mathcal{E}M^{1/2}}}_{L^2_{x,v}}\norm{\pi_L\pa{\partial_l^{0} f}}_{L^2_{x,v}} 
\\&\leq C_{s,\Lambda}\norm{\pi_L\pa{\partial_l^{0} f}}_{L^2_{x,v}}.
\end{align*}
When there are $v$ derivatives we still have $\partial_l^j\pa{\mathcal{E}M^{1/2}}$ that is a polynomial times a Maxwellian and therefore belongs to $L^2_{x,v}$. Thus
$$\abs{\langle \partial_{l}^j \pa{\mathcal{E}M^{1/2}},\partial_l^{j} f \rangle_{L^2_{x,v}}} \leq C_{s,\Lambda}\norm{\partial_l^{j} f}_{L^2_{x,v}}.$$
Also for similar reasons
$$\abs{\langle \partial_{l+\delta_i}^0 \pa{\mathcal{E}M^{1/2}},\partial_l^{\delta_i} f \rangle_{L^2_{x,v}}} \leq C_{s,\Lambda}\norm{\partial_l^{\delta_i} f}_{L^2_{x,v}}.$$
Those three estimates yield the expected results using the Young inequality.
\end{proof}


\section{\textit{A priori} estimates in Sobolev spaces}\label{a priori deterministic}
We provide here Sobolev estimates for the nonlinear perturbed equation $\eqref{eq:perturbative BE}$. We shall work in twisted Sobolev norms that catch the hypocoercivity of the Boltzmann perturbed linear operator. Indeed, as shown by the estimates on the Boltzmann linear operator $L$, we do have a full negative feedback, and a gain of weight, as soon as $\partial_l^j$ includes one velocity derivative. Unfortunately, the negative feedback offered by $L$ on pure spatial derivative only controls the orthogonal part $\pi_L^\bot$. In the exact same spirit as \cite{MouNeu,Bri3}, a small portion of scalar product between spatial and velocity derivative is added to the standard Sobolev norm in order to take advantage of the commutator
$$[v\cdot\nabla_x,\nabla_v] = -\nabla_x.$$
We shall establish \textit{a priori} estimates in Sobolev space to the perturbed equation $\eqref{eq:perturbative BE}$ that we recall here
\begin{equation}\label{eq:perturbative BE apriori}
\begin{split}
\partial_t f &= - \frac{1}{\eps}v\cdot\nabla_x f - \eps \vec{E_t}(x)\cdot\nabla_vf - \eps \mathcal{E}(t,x,v) f + \frac{1}{\eps^2}L[f] + \frac{1}{\eps}\Gamma[f,f] -2 \mathcal{E}(t,x,v)M^{1/2}
\\&:= \sum\limits_{i=1}^6 S_i(t,x,v).
\end{split}
\end{equation}
We gather the estimates derived in Section \ref{sec:estimates operators} to construct and equivalent Sobolev norm of $f$ that can be controlled As $T_0$, $a$ and $A >0$ have been fixed in Proposition \ref{prop:estimates E ronde} we drop the dependencies on the subscripts. Also, as we shall always work with derivatives of order less than a given $s$, we drop the dependencies on $s$. Note however that a lot of different parameters are involved and so, to avoid any loop in their later choice, we will index the constants with these parameters, even if it complicates the reading: the important dependencies are $\Lambda$ and $\eta_i$.
\par We shall address the velocity derivatives and the pure spatial derivatives at different orders in $\eps$, in the spirit of \cite{Bri3}. In what follows we shall use the notation
\begin{equation}\label{eq:Sobolev eps}
\forall f \in H^s_{x,v},\quad \norm{f}^2_{\mathcal{H}^{s}_{\eps}} =  \sum\limits_{\abs{l}\leq s}\norm{\partial_l^0 f}^2_{L^2_{x,v}} + \eps^2 \sum\limits_{\underset{\abs{j}\geq 1}{\abs{j}+ \abs{l}\leq s}}\norm{\partial_l^jf}^2_{L^2_{x,v}}.
\end{equation}


\subsection{Estimates for spatial derivatives}\label{subsec:x derivative deterministic}

As mentioned at the beginning of the present section the pure $x$-derivatives in Sobolev spaces must be handled with the help of the transport operator. We thus define
\begin{equation}\label{eq:Qli}
\forall l \in \N^d,\:\forall 1\leq i\leq d,\quad Q_{l,i}(f) = p \norm{\partial_{l+\delta_i}^0f}^2_{L^2_{x,v}}+ q \eps^2\norm{\partial_l^{\delta_i} f}^2_{L^2_{x,v}} +  \eps r\langle \partial_{l+\delta_i}^0f , \partial_l^{\delta_i} f\rangle_{L^2_{x,v}}.
\end{equation}
The numbers $p$, $q$ and $r$ are constants that we shall define later and select to ensure that $Q_{l,i}(f)$ is a norm equivalent to its standard Sobolev counterpart. Before getting a full Sobolev estimate, we first study the term $Q_{l,i}$. The crucial idea being that the terms $\norm{f}^2_{L^2_{x,v}}$ arising from $S_2$ will be controlled by the fluctuation of the characteristics (\textit{i.e.} the gain due to $S_3$), rather than by the negative feedback of the Boltzmann operator, whilst the source term $S_6$ will find itself controlled by the latter. In what follows our Propositions are divided into two different cases: $\eps=1$ and $\eps<1$. The difference here is that, in the case $\eps =1$, we must keep the negative feedback brought by the fluctuation, whereas for small $\eps$ we can discard it (see Remark~\ref{rem:difference eps=1}).

\begin{prop}\label{prop:estimate Qli}
Let $s$ be in $\N$ and $l$ be a multi-index such that $\abs{l}=s$. There exist $0<\eps_s\leq 1$ for which we have the following results.
\par\textbf{Case $\mathbf{0<\boldsymbol\eps \leq \eps_s}$.} There exists $p$, $q$, $r$ and $C_0>0$, such that
$$\quad Q_{l,i}(\cdot) \sim \norm{\partial_{l+\delta_i}^0\cdot}^2_{L^2_{x,v}}+ \eps^2\norm{\partial_l^{\delta_i} \cdot}^2_{L^2_{x,v}}$$
and if $f$ is a solution to the perturbative equation $\eqref{eq:perturbative BE apriori}$, then 
\begin{equation*}
\begin{split}
\forall t \in [0,T_0), \quad\frac{d}{dt}Q_{l,i}(f) \leq & - 2\pa{\frac{1}{\eps^2}\norm{\pi_L^\bot\pa{\partial_{l+\delta_i}^0 f}}^2_{L^2_\gamma} +\norm{\partial_{l+\delta_i}^0f}^2_{L^2_{x,v}}+  \norm{\partial_{l}^{\delta_i} f}^2_{L^2_\gamma}}
\\&+  \eps^{1-\eu}C_0\sum\limits_{1\leq j,k\leq d}\norm{\partial^{\delta_k}_{l+\delta_i-\delta_j}f}^2_{L^2_\gamma} + C_0\eps^2\sum\limits_{1\leq j,k\leq d}\norm{\partial^{\delta_i + \delta_k}_{l-\delta_j}f}^2_{L^2_\gamma}
\\&+ C_0\cro{\norm{f}^2_{\mathcal{H}^{s+1}_{\eps}}\norm{f}^2_{H^{s+1}_\gamma} + \norm{f}^2_{H^s_{x,v}}+ 1}.
\end{split}
\end{equation*}

\textbf{Case $\mathbf{\boldsymbol\eps = 1}$.} There exists $p$, $q$, $r$ and $C_0>0$ such that
$$\quad Q_{l,i}(\cdot) \sim \norm{\partial_{l+\delta_i}^0\cdot}^2_{L^2_{x,v}}+ \norm{\partial_l^{\delta_i} \cdot}^2_{L^2_{x,v}}$$
and if $f$ is a solution to the perturbative equation $\eqref{eq:perturbative BE apriori}$, then for any $\Lambda > 1$,
\begin{equation*}
\begin{split}
\forall t \in [0,T_0), \quad\frac{d}{dt}Q_{l,i}(f) \leq & -2\lambda_{s,\Lambda}\pa{\norm{\pi_L^\bot\pa{\partial_{l+\delta_i}^0 f}}^2_{L^2_\gamma} +\norm{\partial_{l+\delta_i}^0f}^2_{L^2_{x,v}}+  \norm{\partial_{l}^{\delta_i} f}^2_{L^2_\gamma}}
\\&-\Lambda\norm{\partial_{l+\delta_i}^0f}^2_{L^2_2}- \Lambda\norm{\partial_{l}^{\delta_i}f}^2_{L^2_2} 
\\&+  C_0\sum\limits_{1\leq j,k\leq d}\norm{\partial^{\delta_k}_{l+\delta_i-\delta_j}f}^2_{L^2_2} + C_0\sum\limits_{1\leq j,k\leq d}\norm{\partial^{\delta_i + \delta_k}_{l-\delta_j}f}^2_{L^2_2}
\\&+ C_{0,\Lambda}\cro{\norm{f}^2_{\mathcal{H}^{s+1}_{\eps}}\norm{f}^2_{H^{s+1}_\gamma} + \norm{f}^2_{H^s_{x,v}}+ 1}.
\end{split}
\end{equation*}
All the constants involved depend explicitly on $T_0$, $s$ and $E$.
\end{prop}

\begin{proof}[Proof of Proposition \ref{prop:estimate Qli}]
We recall that $f$ is solution to
$$\partial_t f  = \sum\limits_{j=1}^6 S_j(t,x,v)$$
which directly implies that
$$\frac{d}{dt}Q_{l,i}(f) = \sum\limits_{j=1}^6 p\langle \partial_{l+\delta_i}^0S_j, \partial_{l+\delta_i}^0 f \rangle_{L^2_{x,v}} + \eps^2 q\langle \partial_{l}^{\delta_i}S_j, \partial_{l}^{\delta_i} f \rangle_{L^2_{x,v}} + \eps r\langle \partial_{l+\delta_i}^0S_j, \partial_{l}^{\delta_i} f \rangle_{L^2_{x,v}}.$$
We therefore directly apply Propositions \ref{prop:properties L} and \ref{prop:estimates Gamma} to control $S_4$ and $S_5$, whereas we use Lemmas \ref{lem:S1}, \ref{lem:S2}, \ref{lem:S3} and \ref{lem:S6} for the other terms. It yields
\begin{equation}\label{eq:Qli general}
\begin{split}
\frac{1}{2}\frac{d}{dt}Q_{l,i}(f) \leq & \frac{C_{\eta_0}r +  p\eta_0 - p}{\eps^2}\lambda_{0,\Lambda}\norm{\pi_L^\bot\pa{\partial_{l+\delta_i}^0 f}}^2_{L^2_\gamma} 
\\&+\cro{\frac{q\eps^2}{\eta_1}+\eta_4p+ \eps^{1-\eu}\frac{1_{\eps <1}}{4}p+\eps\frac{1_{\eps <1}}{4}r - r}\norm{\partial_{l+\delta_i}^0f}^2_{L^2_{x,v}}
\\&+ \cro{q(\eta_1+\eta'_0+\eta_4+\eps^{3-e}\frac{1_{\eps <1}}{4})+r(\eta_0+\eta_4)\lambda_{0,\Lambda}+\eps^{2-e}\frac{1_{\eps <1}}{4}r - q\lambda_{s,\Lambda}} \norm{\partial_{l}^{\delta_i} f}^2_{L^2_\gamma}
\\&+ \eps^{1+\eu}\Lambda\cro{C_{\eta_3}r + \eta_2 p-\frac{p}{2}}\norm{\partial_{l+\delta_i}^0f}^2_{L^2_2}
\\&+ \eps^{3+e}\Lambda\cro{\eta_2 q + (\eta_2 + \eta_3) r  - \frac{q}{2}}\norm{\partial_{l}^{\delta_i}f}^2_{L^2_2} 
\\&+ \eps^{1-\eu}\cro{pC_{\eta_2}+ d \frac{C_{\eta_2}}{\Lambda}r}\sum\limits_{1\leq j,k\leq d}\norm{\partial^{\delta_k}_{l+\delta_i-\delta_j}f}^2_{L^2_{x,v}} 
\\&+ \cro{\eps^2 q C_{\eta_2}}\sum\limits_{1\leq j,k\leq d}\norm{\partial^{\delta_i + \delta_k}_{l-\delta_j}f}^2_{L^2_{x,v}}
\\&+ C_{p,q,r,\Lambda,\eta}\cro{\norm{f}^2_{\mathcal{H}^{s+1}_{\eps}}\norm{f}^2_{H^{s+1}_\gamma} + \norm{f}^2_{H^s_{x,v}}+ 1}.
\end{split}
\end{equation}
We firstly emphasize that the twice indexed sums are a cruder estimate than the one we actually derived in the Lemmas: we added some terms that were formerly absent, but we think it provides a better reading. We secondly emphasize that the estimate on the bilinear term $S_5$ in Proposition \ref{prop:estimates Gamma} gives a control of the form $\norm{f}^2_{H^{s+1}_{x,v}}\norm{f}^2_{H^{s+1}_\gamma}$ which translates into a control of the form $\norm{f}^2_{\mathcal{H}^{s+1}_{\eps}}\norm{f}^2_{H^{s+1}_\gamma}$ when multiplying by powers of $\eps$. We lastly used $\norm{\pi_L\pa{\partial^0_l f}}_{L^2_{x,v}} \leq \norm{\partial^0_l f}_{L^2_{x,v}}$ when applying Lemma \ref{lem:S6}.
\par In what follows, we recall that $0<\eps \leq 1$. We shall now choose the constants carefully, which is why we indexed all generic constant by their dependencies in order to avoid any loop.
\begin{remark}\label{rem:difference eps=1}
The choices are different for $\eps =1$ or any $\eps<1$ because the control of the specific term $\mathcal{T}:=\sum\limits_{1\leq j,k\leq d}\norm{\partial^{\delta_k}_{l+\delta_i-\delta_j}f}^2_{L^2_{x,v}}$ will be achieved in two different ways. For $\eps=1$, the negative feedback of the fluctuation $\norm{\partial^{\delta_i}_lf}_{L^2_2}$ can control these terms, taking $\Lambda$ sufficiently large. Such an approach does not work for general values of $\eps$ (because we control $v$-derivatives with a degenerate weight $\eps^2$). In the general case therefore, the term $\mathcal{T}$ will be absorbed by the negative feedback that the linear Boltzmann operator provides, and this latter approach requires a sufficiently small value of $\eps$.
\par Note that this distinction is quite artificial, and is due to our choice of Sobolev norm with coefficient $\eps^2 \norm{\partial_{l}^{\delta_i}\cdot}^2_{L^2_{x,v}}$. Working with this weighted norm facilitates various computations and estimates, but \cite{Bri3} showed that a finer norm, which is not degenerating when $\eps$ tends to $0$, can catch the hypocoercivity of the Boltzmann linear operator. With more technicalities, we think we could use the latter norm to avoid this splitting into two regimes, and always control the problematic term with the negative feedback generated by the fluctuations for any $\eps$.
\end{remark}
We start with the case $\eps <1$ and fix the following quantities:
\begin{enumerate}
\item $\Lambda=1$;
\item $q$ sufficiently large such that $-q\lambda_{s,1} \leq -5$;
\item $\eta_1$ small enough such that $q\eta_1 \leq 1$;
\item $r$ sufficiently large such that $\frac{q}{\eta_1} - \frac{3r}{4} \leq -4$;
\item $\eta_0=\eta_0'$  small enough such that 
	\begin{itemize}
	\item $2\eta_0 - \lambda_{s,1} \leq -\frac{\lambda_{s,1}}{2}$,
	\item $2\eta_0q + 2\eta_0\lambda_{0,1}r \leq -1$;
	\end{itemize}
\item $\eta_2$ small enough such that $\eta_2 \leq 1/8$ and $\eta_2 r \leq q/8$;
\item $\eta_3$ small enough such that $\eta_3 r \leq q /8$;
\item $p$ sufficiently large such that
	\begin{itemize}
	\item $C_{1,\eta_0} r + p \eta_0 - p \leq -1$,
	\item $C_{\eta_3}r + \eta_2 p-\frac{p}{2} \leq 0$,
	\item $r^2\leq pq$ and $q\leq p$ so that $Q_{l,i}(\cdot) \sim \norm{\partial_{l+\delta_i}^0\cdot}^2_{L^2_{x,v}}+ \eps^2\norm{\partial_l^{\delta_i} \cdot}^2_{L^2_{x,v}}$
	\end{itemize}
\item $\eta_4 \leq \eta_0$ small enough such that $\eta_4 p \leq 1$ - note that point this allows to use point (5) above;
\item At last we need $\eps$ small enough such that $\eps^{1-\eu}\frac{p}{4}\leq 1,\:\eps^{3-\eu}\frac{q}{4},\:\eps^{2-\eu}\frac{r}{4} \leq 1$.
\end{enumerate}
Such a choice yields exactly the expected result for $0<\eps<1$.

\begin{remark}\label{rem:e<1}
One clearly sees here that, if $\eu$ were too large,  $\eu\geq 1$, then one could not make $\eps^{1-\eu}\frac{p}{4}$ small as desired. In other terms, in that case of large coefficient $\eu$, the amplitude of the evolution of the $v$ characteristics would not be compensated by the gain due, via hypocoercive estimates, to the free transport in $x$.
\end{remark}

Now let us deal with the particular case $\eps=1$. The crucial step will be to fix the constants $p$, $r$ and $\eta_2$ to ensure that the term
$\sum\limits_{1\leq j,k\leq d}\norm{\partial^{\delta_k}_{l+\delta_i-\delta_j}f}^2_{L^2_{x,v}}$ has a multiplicative constant independent of $\Lambda$. Then $\Lambda$ should precisely be chosen large enough to absorb these contributions. In order to achieve this goal we transform $\eqref{eq:Qli general}$ by estimating
$$q\eta_1 \norm{\partial_{l}^{\delta_i} f}^2_{L^2_\gamma} \leq Cq\eta_1 \norm{\partial_{l}^{\delta_i} f}^2_{L^2_2} \quad\mbox{and}\quad 2r\eta_0\lambda_{0,\Lambda} \norm{\partial_{l}^{\delta_i} f}^2_{L^2_\gamma} \leq Cr \eta_0  \norm{\partial_{l}^{\delta_i} f}^2_{L^2_2},$$
since $\lambda_{0,\Lambda} \leq \lambda_{0,0}$ from the proof of Proposition \ref{prop:properties L}. We infer, for $\Lambda \geq 1$, the estimate
\begin{equation*}
\begin{split}
\frac{d}{dt}Q_{l,i}(f) \leq & \cro{C_{\eta_0}r + p\eta_0 - p}\lambda_{0,\Lambda}\norm{\pi_L^\bot\pa{\partial_{l+\delta_i}^0 f}}^2_{L^2_\gamma} +\cro{\frac{q}{\eta_1}+\eta_4 p - r}\norm{\partial_{l+\delta_i}^0f}^2_{L^2_{x,v}}
\\&+ \cro{2q\eta'_0 - q\lambda_{s,\Lambda}} \norm{\partial_{l}^{\delta_i} f}^2_{L^2_\gamma}
\\&+ \Lambda\cro{C_{\eta_3}r + \eta_2 p-\frac{p}{2}}\norm{\partial_{l+\delta_i}^0f}^2_{L^2_2}
\\&+ \Lambda\cro{(C\eta_1+\eta_2+\eta_4) q + (C\eta_0+\eta_2 + \eta_3+\eta_4) r  - \frac{q}{2}}\norm{\partial_{l}^{\delta_i}f}^2_{L^2_2} 
\\&+ \cro{ pC_{\eta_2}+ d C_{\eta_2}r}\sum\limits_{1\leq j,k\leq d}\norm{\partial^{\delta_k}_{l+\delta_i-\delta_j}f}^2_{L^2_{x,v}} + q C_{\eta_2}\sum\limits_{1\leq j,k\leq d}\norm{\partial^{\delta_i + \delta_k}_{l-\delta_j}f}^2_{L^2_{x,v}}
\\&+ C_{p,q,r,\Lambda,\eta}\cro{\norm{f}^2_{\mathcal{H}^{s+1}_{\eps}}\norm{f}^2_{H^{s+1}_\gamma} + \norm{f}^2_{H^s_{x,v}}+ 1}.
\end{split}
\end{equation*}
We can now choose our different constants in the following way:
\begin{enumerate}
\item $q = 8$;
\item $\eta_1$ small enough such that $Cq\eta_1 \leq 1$;
\item $r$ sufficiently large such that $\frac{q}{\eta_1} - r \leq -3$;
\item $\eta_0$, $\eta_2$ and $\eta_3$ small enough such that 
	\begin{itemize}
	\item $\eta_0 \leq \frac{1}{4}$,
	\item $C\eta_1 + 2\eta_2 \leq 1$,
	\item $(C\eta_0 + 2\eta_2 + \eta_3)r \leq 1$,
	\item $\eta_2 \leq \frac{1}{4}$
	\end{itemize}
\item $\eta_0' = \eta_0'(\Lambda)$ small enough such that $3\eta_0' \leq -\lambda_{s,\Lambda}$
\item  $p$ sufficiently large such that
	\begin{itemize}
	\item $C_{\eta_0} r +2 p\eta_0 - p \leq -1$,
	\item $C_{\eta_3}r + \eta_2 p-\frac{p}{2} \leq -1$,
	\item $r^2\leq pq$ and $q\leq p$ so that $Q_{l,i}(\cdot) \sim \norm{\partial_{l+\delta_i}^0\cdot}^2_{L^2_{x,v}}+ \eps^2\norm{\partial_l^{\delta_i} \cdot}^2_{L^2_{x,v}}$.
	\end{itemize}
\item At last $\eta_4 \leq \eta_2$ small enough such that $\eta_4 p \leq 1$ - note that point this allows to use point (5) above;.
\end{enumerate}
These choices yield the expected result, emphasizing that we manage to choose $p$, $q$, $r$ and $C_{\eta_2}$ independently of $\Lambda$.
\end{proof}


\subsection{Estimates for velocity derivatives}\label{subsec:v derivative deterministic}

We now turn to the terms that include velocity derivatives for which the linear Boltzmann operator provides a full negative feedback (this is the second term in \eqref{GainL}).

\begin{prop}\label{prop:estimate djl}
Let $s$ be in $\N^*$ and $l$ and $j$ be multi-indexes with  $\abs{l}+\abs{j}=s+1$ with $\abs{j}\geq 2$. There exist $0<\eps_s\leq 1$
\par\textbf{Case $\mathbf{0<\boldsymbol\eps \leq \eps_s}$.} There exists $\lambda_{s,1}$, $C_1 >0$ such that if $f$ is a solution to the perturbative equation $\eqref{eq:perturbative BE apriori}$ then
\begin{equation*}
\begin{split}
\forall t \in [0,T_0),\quad \frac{d}{dt}\norm{\partial_l^jf}^2_{L^2_{x,v}} \leq& -\frac{\lambda_{s,1}}{\eps^2}\norm{\partial^j_lf}^2_{L^2_\gamma} + C_1\sum\limits_{1\leq i,k\leq d}\norm{\partial^{j+\delta_k}_{l-\delta_i}f}^2_{L^2_\gamma} 
\\&+ C_1\sum\limits_{k=1}^d\norm{\partial^{j-\delta_k}_{l+\delta_k}f}^2_{L^2_{x,v}}
\\&+ \frac{C_1}{\eps^2}\cro{\norm{f}^2_{\mathcal{H}^{s+1}_{\eps}}\norm{f}^2_{H^{s+1}_\gamma} + \norm{f}^2_{H^s_{x,v}}+ 1}.
\end{split}
\end{equation*}

\textbf{Case $\mathbf{\boldsymbol\eps = 1}$.} There exists $C_1 >0$ such that if $f$ is a solution to the perturbative equation $\eqref{eq:perturbative BE apriori}$ then for any $\Lambda \geq 1$,
\begin{equation*}
\begin{split}
\forall t \in [0,T_0),\quad \frac{d}{dt}\norm{\partial_l^jf}^2_{L^2_{x,v}} \leq& -\lambda_{s,\Lambda}\norm{\partial^j_lf}^2_{L^2_\gamma} + C_1\sum\limits_{k=1}^d\norm{\partial^{j-\delta_k}_{l+\delta_k}f}^2_{L^2_{x,v}}
\\&-\Lambda\norm{\partial_l^jf}^2_{L^2_2} +C_1\sum\limits_{1\leq i,k\leq d}\norm{\partial^{j + \delta_k}_{l-\delta_i}f}^2_{L^2_2}
\\&+ C_{1,\Lambda}\cro{\norm{f}^2_{\mathcal{H}^{s+1}_{\eps}}\norm{f}^2_{H^{s+1}_\gamma} + \norm{f}^2_{H^s_{x,v}}+ 1}.
\end{split}
\end{equation*}
All the constants depend explicitly on $s$, $E$ and $T_0$.
\end{prop}

\begin{proof}[Proof of Proposition \ref{prop:estimate djl}]
As for $Q_{l,i}$, we recall that $f$ is solution to
$$\partial_t f  = \sum\limits_{j=1}^6 S_j(t,x,v)$$
so we directly apply  Propositions \ref{prop:properties L} and \ref{prop:estimates Gamma} to control $S_4$ and $S_5$, whereas we use Lemmas \ref{lem:S1}, \ref{lem:S2}, \ref{lem:S3} and \ref{lem:S6} for the other terms. We have
\begin{equation}\label{eq:djl general}
\begin{split}
\frac{1}{2}\frac{d}{dt}\norm{\partial_l^jf}^2_{L^2_{x,v}} \leq & \frac{\eta_0+\eta_0'+\eta_1+\eta_4+\eps^{3-e}\frac{\mathbf{1}_{\eps <1}}{4}-\lambda_{s,\Lambda}}{\eps^2}\norm{\partial_l^jf}^2_{L^2_\gamma} 
\\&+ \cro{\eps^{1+\eu}\eta_2\Lambda - \eps^{1+\eu} \Lambda}\norm{\partial_l^jf}^2_{L^2_2}
\\&+ \frac{1}{\eta_1}\sum\limits_{k=1}^d\norm{\partial^{j-\delta_k}_{l+\delta_k}f}^2_{L^2_{x,v}} +\eps^{1-\eu} C_{\eta_2}\sum\limits_{1\leq i,k\leq d}\norm{\partial^{j+\delta_k}_{l-\delta_i}f}^2_{L^2_{x,v}} 
\\&+ \frac{C_{1,\Lambda}}{\eps^2}\cro{\norm{f}^2_{\mathcal{H}^{s+1}_{\eps}}\norm{f}^2_{H^{s+1}_\gamma} + \norm{f}^2_{H^s_{x,v}}+ 1}.
\end{split}
\end{equation}

Here again our choice of constant will differ if $\eps=1$. First let us consider the general case $0<\eps <1$. We take
\begin{enumerate}
\item $\Lambda=1$;
\item $\eta_0=\eta_0'=\eta_1=\eta_4$ small enough such that 
	\begin{itemize}
	\item $\eta_0+\eta_0'+\eta_1-\lambda_{s,1} \leq -\lambda_{s,1}/2$,
	\item $\eta_2 \leq 1/2$;
	\end{itemize}
\item $\eps$ sufficiently small such that $\frac{\eps^{3-e}}{4}\leq \frac{\lambda_{s,1}}{4}.$
\end{enumerate}
and these choices lead to the expected estimate.

The particular case $\eps=1$ is dealt with differently and we modify $\eqref{eq:djl general}$ by estimating
$$\eta_1\norm{\partial_l^jf}^2_{L^2_\gamma} \leq \eta_1\norm{\partial_l^jf}^2_{L^2_2},$$
to obtain
\begin{equation*}
\begin{split}
\frac{1}{2}\frac{d}{dt}\norm{\partial_l^jf}^2_{L^2_{x,v}} \leq & \cro{\eta_0+\eta_0'-\lambda_{s,\Lambda}}\norm{\partial_l^jf}^2_{L^2_\gamma} + \cro{2\eta_1+\eta_2\Lambda - \Lambda}\norm{\partial_l^jf}^2_{L^2_2}
\\&+ \frac{1}{\eta_1}\sum\limits_{1\leq k\leq d}\norm{\partial^{j-\delta_k}_{l+\delta_k}f}^2_{L^2_{x,v}} + C_{\eta_2}\sum\limits_{1\leq i,k\leq d}\norm{\partial^{j + \delta_k}_{l-\delta_i}f}^2_{L^2_{x,v}}
\\&+ C_{1,\Lambda,\eta}\cro{\norm{f}^2_{\mathcal{H}^{s+1}_{\eps}}\norm{f}^2_{H^{s+1}_\gamma} + \norm{f}^2_{H^s_{x,v}}+ 1}.
\end{split}
\end{equation*}
Taking $\Lambda \geq 1$ we can choose the constants $\eta_1$ and $\eta_2$ independently of $\Lambda$ in the same manner as for the general case and obtain the expected result.
\end{proof}


\subsection{Estimates for the full Sobolev norm}\label{subsec:full sobolev deterministic}

We now gather the previous estimates to establish a full control over the twisted Sobolev norm. We start with the Sobolev norm corresponding to a fixed total number of derivatives, and then, layer by layer, deal with the norm that estimates all the derivatives.

\begin{prop}\label{prop:sobolev dot deterministic}
Let $s\geq s_0$ (where $s_0$ is given in Proposition~\ref{prop:estimates Gamma}) and $1\leq s' \leq s$. Then there exists $\Lambda>0$, $\lambda_{s'}>0$, $C_{s'}>0$ and $0<\eps_s {s'}\leq 1$ such that for ($\eps =1$ or $0<\eps\leq \eps_{s'}$) there exists a functional $F_{s'}$ such that
$$F_{s'} \sim \sum\limits_{\abs{l}={s'}}\norm{\partial^0_l \cdot}^2_{L^2_{x,v}} + \eps^2 \sum\limits_{\underset{\abs{j}\geq 1}{\abs{l}+\abs{j}={s'}}}\norm{\partial_l^j \cdot}^2_{L^2_{x,v}}$$
and if $f$ is a solution to the perturbative equation $\eqref{eq:perturbative BE apriori}$ then for any $\Lambda >0$, for all $t\in[0,T_0)$,
\begin{equation*}
\begin{split}
\frac{d}{dt}F_{s'}(f) \leq &-\lambda\pa{ \frac{1}{\eps^2}\sum\limits_{\abs{l}={s'}}\norm{\pi_L^\bot\pa{\partial^0_l f}}^2_{L^2_\gamma} +  \sum\limits_{\underset{\abs{j}\geq 1}{\abs{l}+\abs{j}={s'}}}\norm{\partial_l^j f}^2_{L^2_\gamma}+\sum\limits_{\abs{l}={s'}}\norm{\partial^0_l f}^2_{L^2_{x,v}}}
\\&+ C_{s'}\cro{\norm{f}^2_{\mathcal{H}^{s}_{\eps}}\norm{f}^2_{H^{s}_\gamma} + \norm{f}^2_{H^{s'-1}_{x,v}}+ 1}.
\end{split}
\end{equation*}
All the constants depend on $s$, $E$ and $T_0$.
\end{prop}

%
\begin{proof}[Proof of Proposition \ref{prop:sobolev dot deterministic}]
We present the proof in two different cases: $\eps$ sufficiently small first and then $\eps=1$. The technicalities are identical but the absorbtion mechanisms are different as explained in Remark \ref{rem:difference eps=1}. Consider some constant $B_{j,l}>0$ to be fixed later, and define the functional
$$F_{s'}(f) = \sum\limits_{\abs{l}={s'}}\sum\limits_{i=1}^dQ_{l,i}(f) + \eps^2\sum\limits_{\underset{\abs{j}\geq 2}{\abs{l}+\abs{j}={s'}}} B_{j,l}\norm{\partial_l^jf}^2_{L^2_{x,v}}.$$
By Proposition \ref{prop:estimate Qli}, we know that, for any $B>0$:
$$F_s \sim \sum\limits_{\abs{l}=s'}\norm{\partial^0_l \cdot}^2_{L^2_{x,v}} + \eps^2 \sum\limits_{\underset{\abs{j}\geq 1}{\abs{l}+\abs{j}={s'}}}\norm{\partial_l^j \cdot}^2_{L^2_{x,v}}.$$

\textbf{Case $\eps$ sufficiently small.} Using directly Proposition \ref{prop:estimate Qli} and Proposition \ref{prop:estimate djl} we see that, for $B_{j,l}=B$ sufficiently large, we have a constant $C_B>0$ independent of $\eps$ such that
\begin{equation*}
\begin{split}
\frac{d}{dt}F_{s'}(f) \leq& -\pa{\frac{1}{\eps^2}\sum\limits_{\abs{l}=s'}\norm{\pi_L^\bot\pa{\partial^0_l f}}^2_{L^2_\gamma} +  \sum\limits_{\underset{\abs{j}\geq 1}{\abs{l}+\abs{j}=s'}}\norm{\partial_l^j f}^2_{L^2_\gamma}+\sum\limits_{\abs{l}=s'}\norm{\partial^0_l f}^2_{L^2_{x,v}}}
\\&+\eps C_B\sum\limits_{\abs{l}+\abs{j}=s'}\sum\limits_{1\leq i,k \leq d} \norm{\partial^{j+\delta_k}_{l-\delta_i}f}^2_{L^2_\gamma}+\norm{\partial^{j-\delta_k}_{l+\delta_j}f}^2_{L^2_{x,v}}
\\&+C_{1,B}\cro{\norm{f}^2_{\mathcal{H}^{s}_{\eps}}\norm{f}^2_{H^{s}_\gamma} + \norm{f}^2_{H^{s'-1}_{x,v}}+ 1}.
\end{split}
\end{equation*}
Therefore taking $\eps$ sufficiently small allows to absorb the second line with the ne\-ga\-ti\-ve first term. 
This is the expected result.

\textbf{Case $\eps=1$.} The proof is exactly the same. Gathering Proposition \ref{prop:estimate Qli} and Proposition \ref{prop:estimate djl} with $B_{j,l}<1$, we infer that there exists $C>0$ independent of $\Lambda$ and $B_{j,l}$ such that
\begin{equation*}
\begin{split}
\frac{d}{dt}F_{s'}(f) \leq& -\lambda_{s',\Lambda}\Big(\sum\limits_{\abs{l}=s'}\norm{\pi_L^\bot\pa{\partial^0_l f}}^2_{L^2_\gamma}+\sum\limits_{\abs{l}=s'}\norm{\partial^0_l f}^2_{L^2_{x,v}}
\\&\quad\quad\quad\quad\quad+  \sum\limits_{\underset{\abs{j}\geq 2}{\abs{l}+\abs{j}=s'}}B_{j,l}\norm{\partial_l^j f}^2_{L^2_\gamma} + \sum\limits_{\abs{l}=s'-1}\norm{\partial^{\delta_i}_lf}^2_{L^2_\gamma}\Big)
\\&+C_1\sum\limits_{\underset{\abs{j}\geq 2}{\abs{j}+\abs{l}=s'}}B_{j,l}\sum\limits_{k=1}^d\norm{\partial^{j-\delta_k}_{l+\delta_k}f}^2_{L^2_\gamma}
\\&-\Lambda\sum\limits_{\abs{j}+\abs{l}=s'}\norm{\partial_l^jf}^2_2
\\&+C\pa{\sum\limits_{\abs{l}+\abs{j}=s'}\sum\limits_{1\leq i,k \leq d} \norm{\partial^{j+\delta_k}_{l-\delta_i}f}^2_{L^2_2}+\norm{\partial^{j-\delta_k}_{l+\delta_j}f}^2_{L^2_2}}
\\&+C\cro{\norm{f}^2_{\mathcal{H}^{s}_{\eps}}\norm{f}^2_{H^{s}_\gamma} + \norm{f}^2_{H^{s'-1}_{x,v}}+ 1}.
\end{split}
\end{equation*}
We can then choose $B_{j,l}=B_{j,l}(\Lambda)$ sufficiently small hierarchically to absorb 
$$C_1\sum\limits_{\underset{\abs{j}\geq 2}{\abs{j}+\abs{l}=s'}}B_{j,l}\sum\limits_{k=1}^d\norm{\partial^{j-\delta_k}_{l+\delta_k}f}^2_{L^2_\gamma}$$
 inside the full negative terms on the first line of the estimate. As $C$ is independent of $\Lambda$ (due to the fact that we could choose the $B_{j,l}<1$) we can at last fix $\Lambda$ sufficiently large such that the $L^2_2$-norms give a negative contribution. This concludes the proof.
\end{proof}

We finally have all the tools to perform a full $H^s_{x,v}$ estimate. The following proposition indeed shows that our choice of perturbative regime compensate the modification of the characteristics due to the presence of the external force even on the fluid part of the solution $\pi_L(f)$.

\begin{prop}\label{prop:sobolev deterministic}
Let $s\geq s_0$ (where $s_0$ is given in Proposition~\ref{prop:estimates Gamma}). Then there exists $\Lambda>0$, $C>0$ and $0<\eps_s \leq 1$ such that, for $\eps =1$ or $0<\eps\leq \eps_s$, there exists a functional 
\begin{equation}\label{normweight}
\norm{\cdot}^2_{\mathcal{H}^s_{x,v}} \sim \sum\limits_{\abs{l}\leq s}\norm{\partial^0_l \cdot}^2_{L^2_{x,v}} + \eps^2 \sum\limits_{\underset{\abs{j}\geq 1}{\abs{l}+\abs{j}\leq s}}\norm{\partial_l^j \cdot}^2_{L^2_{x,v}}
\end{equation}
such that, if $f$ is a solution to the perturbative equation $\eqref{eq:perturbative BE apriori}$ with initial data $f_\mathrm{in}$ satisfying
$$\abs{\int_{\T^d\times\R^d}\pa{\begin{array}{c} 1 \\ v \\ \abs{v}^2\end{array}}f_\mathrm{in}(x,v)\sqrt{M}|_{t=0}\:dxdv} \leq C_\mathrm{in} \eps^{\eu},$$
then
\begin{equation}\label{eq:final estimate}
\begin{split}
\forall t\in [0,T_0),\quad \frac{d}{dt}\norm{f}^2_{\mathcal{H}^s_{x,v}} \leq &-(\lambda-C_s\norm{f}^2_{\mathcal{H}^{s}_{\eps}})\norm{f}^2_{H^s_\gamma} + C_s.
\end{split}
\end{equation}
All the constants depend on $s$, $E$ and $T_0$, but are independent of $\eps$.
\end{prop}

\begin{proof}[Proof of Proposition \ref{prop:sobolev deterministic}]
As we shall see, the proof follows directly from Proposition \ref{prop:sobolev dot deterministic} apart from $s=1$ which has to be treated more carefully due to the lack of full negative return on the $L^2_{x,v}$-norm. Fixing $s$ in $\N^*$ we use Proposition \ref{prop:sobolev dot deterministic} to construct
\begin{multline*}
\frac{d}{dt}\sum\limits_{s'=2}^s \alpha_{s'}F_s(f) \\
\leq -\lambda \sum\limits_{s'=1}^s \pa{ \frac{1}{\eps^2}\sum\limits_{\abs{l}=s'}\norm{\pi_L^\bot\pa{\partial^0_l f}}^2_{L^2_\gamma} +  \sum\limits_{\underset{\abs{j}\geq 1}{\abs{l}+\abs{j}=s'}}\norm{\partial_l^j f}^2_{L^2_\gamma}+\sum\limits_{\abs{l}=s}\norm{\partial^0_l f}^2_{L^2_{x,v}}}
\\+ C \sum\limits_{1\leq s' \leq s}\alpha_{s'}\norm{f}^2_{H^{s'-1}_{x,v}}+ C_s\cro{\norm{f}^2_{\mathcal{H}^{s}_{\eps}}\norm{f}^2_{H^{s}_\gamma} +1},
\end{multline*}
where $\alpha_{s'}$ are constants that we choose sufficiently small hierarchically (from $s$ to $1$) in order to absorb the first positive term at rank $s'$ on the right-hand side by the negative feedback at rank $s'-1$. We infer for all $t$ in $[0,T_0)$,
\begin{multline}\label{eq:start deterministic}
\frac{d}{dt}\sum\limits_{s'=2}^s \alpha_{s'}F_s(f)\\
 \leq -\lambda \sum\limits_{s'=1}^s \pa{ \frac{1}{\eps^2}\sum\limits_{\abs{l}=s'}\norm{\pi_L^\bot\pa{\partial^0_l f}}^2_{L^2_\gamma} +  \sum\limits_{\underset{\abs{j}\geq 1}{\abs{l}+\abs{j}=s'}}\norm{\partial_l^j f}^2_{L^2_\gamma}+\sum\limits_{\abs{l}=s}\norm{\partial^0_l f}^2_{L^2_{x,v}}}
\\+  \norm{f}^2_{H^1_{x,v}} + C_s\cro{\norm{f}^2_{\mathcal{H}^{s}_{\eps}}\norm{f}^2_{H^{s}_\gamma} +1}
\end{multline}
and we have from Proposition \ref{prop:sobolev dot deterministic}:
$$\sum\limits_{1\leq s'\leq s} \alpha_{s'}F_s(f)\sim \sum\limits_{1\leq s' \leq s}\pa{\sum\limits_{\abs{l}\leq s'}\norm{\partial^0_l \cdot}^2_{L^2_{x,v}} + \eps^2 \sum\limits_{\underset{\abs{j}\geq 1}{\abs{l}+\abs{j}\leq s'}}\norm{\partial_l^j \cdot}^2_{L^2_{x,v}}}$$
so that
\begin{multline}\label{eq:start deterministic2}
\frac{d}{dt}\sum\limits_{s'=2}^s \alpha_{s'}F_s(f) \leq -\left(\frac{\lambda}{2}-C_{s}\norm{f}^2_{\mathcal{H}^{s}_{\eps}}\right)\norm{f}^2_{H^s_\gamma} 
+ C_s \norm{f}^2_{H^1_{x,v}} + C_s.
\end{multline}

\bigskip
\textbf{Control of fluid part by spatial derivatives.} A key property of our proof will be to recover the full coercivity on the $L^2_{x,v}$-norm by controlling $\norm{\pi_L(f)}^2_{L^2_{x,v}}$ by $\norm{\nabla_x f}^2_{L^2_{x,v}}$ which is fully coercive. First, as the eigenfunction of $L$ are polynomials times Maxwellian, see $\eqref{eq:ker L}$ and $\eqref{piL}$, we easily obtain (we refer to \cite[equation (3.3)]{Bri3} for a direct proof) that there exists $c_\pi,\:C_\pi >0$ such that
\begin{equation}\label{eq:Cpi}
\forall 0\leq \abs{j}+\abs{l} \leq s, \quad c_\pi \norm{\partial_l^j\pi_L^\bot(f)}^2_{L^2_{x,v}}\leq  \norm{\pi_L\pa{\partial_l^0f}}^2_{L^2_\gamma} \leq C_\pi \norm{\pi_L\pa{\partial_l^0f}}^2_{L^2_{x,v}}.
\end{equation}
So the only part left to estimate is $\norm{\pi_L(f)}^2_{L^2_{x,v}}$. Coming back to \eqref{piL}, we have, for a constant $C$ depending on $a,A,\eu$ and $T_0$,
\[
\norm{\pi_L(f)}^2_{L^2_{x,v}}\leq C\sum\limits_{i=0}^{d+1} \int_{\T^d} |\dual{f}{\bar{\phi}_i\sqrt{M}}|^2 dx.
\]
We use the Poincar\'e-Wirtinger inequality to obtain, for possibly a different constant $C$,
\[
\norm{\pi_L(f)}^2_{L^2_{x,v}}
\leq 
C\sum\limits_{i=0}^{d+1} \left[\int_{\T^d} |\nabla_x\dual{f}{\bar{\phi}_i\sqrt{M}}|^2 dx
+\left|\int_{\T^d}\dual{f}{\bar{\phi}_i\sqrt{M}} dx\right|^2\right],
\]
which gives in turn
\begin{equation}\label{eq:poincare0}
\norm{\pi_L(f)}^2_{L^2_{x,v}}\leq C\norm{\nabla_x(\pi_L(f))}^2_{L^2_{x,v}}+C\int_{\R^d}\left|\int_{\T^d}\pi_L(f) dx\right|^2 dv.
\end{equation}
For the classical Boltzmann equation the preservation of mass, momentum and e\-ner\-gy gives the cancellation $ \int_{\T^d}\pi_L(f) dx=0$, which is no longer satisfied here however.

\par We come back to the original equation $\eqref{eq:BE force}$ on $F = M + \eps \sqrt{M}\:f$:
$$\partial_t F + \frac{1}{\eps}v\cdot\nabla_x F  + \eps \vec{E_t}(x)\cdot \nabla_v F = \frac{1}{\eps}Q(F,F).$$
We multiply this equation by $1$, $v$ and $\abs{v}^2$ and we integrate over $\T^d\times\R^d$. This yields
\begin{align}
\frac{d}{dt}\int_{\T^d\times\R^d} F(t,x_*,v_*)dx_*dv_* &= 0,\label{Moment0}
\\ \frac{d}{dt}\int_{\T^d\times\R^d} v_*F(t,x_*,v_*)dx_*dv_* &= \eps \int_{\T^d\times\R^d}\vec{E_t}(x_*)F(t,x_*,v_*)dx_*dv_*,\label{Moment1}
\\\frac{d}{dt}\int_{\T^d\times\R^d} \abs{v_*}^2F(t,x_*,v_*)dx_*dv_* &= 2\eps \int_{\T^d\times\R^d}\vec{E_t}(x_*)\cdot v_* F(t,x_*,v_*)dx_*dv_*.\label{Moment2}
\end{align}
Note that we used that $Q(f,g)$ is orthogonal to $\mbox{Ker}(L)$ in $L^2_{x,v}$. 
We insert the relation $F = M + \eps \sqrt{M}\:f$ in the previous estimates. We compute and bound for $\delta=0$ or $\delta=2$,
\begin{align}
\abs{\int_{\T^d\times\R^d}\abs{v_*}^\delta\pa{M(0,x_*,v_*) - M(t,x_*,v_*)}} &= \frac{e^{-\eps^{1+\eu}\frac{A}{1+t}}}{(1+\eps^{1+\eu}\frac{a}{1+t})^{\frac{2d+\delta}{4}}} - \frac{e^{-\eps^{1+\eu} A}}{(1+\eps^{1+\eu} a)^{\frac{2d+\delta}{4}}}\nonumber
\\ &\leq \eps^{1+\eu}C_{a,A}.\label{GlobaldtMOK}
\end{align}
From \eqref{Moment0}, \eqref{GlobaldtMOK} and the smallness hypothesis \eqref{HypSmallGlobalMoments}, we deduce that
\begin{equation}\label{Moment0OK}
\sup_{t\in[0,T_0]}\left|\int_{\T^d\times\R^d} f_*(t) \sqrt{M}_*(t)\:dx_*dv_*\right|\leq C\eps^\eu.
\end{equation}
We use the cancellation condition in \eqref{HypE} and \eqref{HypSmallGlobalMoments} to get 
\begin{multline}\label{Moment1OK}
\sup_{t\in[0,T_0]}\left|\int_{\T^d\times\R^d} v_* f_*(t) \sqrt{M}_*(t)\:dx_*dv_*\right|\\
\leq C\eps^\eu
+\eps T_0 C_E \sup_{t\in[0,T_0]}\int_{\T^d}\left|\int_{\R^d} v_* f_*(t) \sqrt{M}_*(t) dv_*\right| dx.
\end{multline}
We use the bounds
\begin{multline*}
\eps T_0 C_E \sup_{t\in[0,T_0]}\int_{\T^d}\left|\int_{\R^d} v_* f_*(t) \sqrt{M}_*(t) dv_*\right| dx\\
\leq \eps T_0 C_E\sup_{t\in[0,T_0]}\norm{\pi_L(f(t))}_{L^2_{x,v}}
\leq C\eps^\eu+\eps^{2(1-\eu)}\sup_{t\in[0,T_0]}\norm{\pi_L(f(t))}_{L^2_{x,v}}^2
\end{multline*}
to obtain
\begin{equation}\label{Moment1OKOK}
\sup_{t\in[0,T_0]}\left|\int_{\T^d\times\R^d} v_* f_*(t) \sqrt{M}_*(t)\:dx_*dv_*\right|\leq C\eps^\eu
+\eps^{2(1-\eu)}\sup_{t\in[0,T_0]}\norm{\pi_L(f(t))}_{L^2_{x,v}}^2.
\end{equation}
A similar procedure gives
\begin{equation}\label{Moment2OK}
\sup_{t\in[0,T_0]}\left|\int_{\T^d\times\R^d} |v_*|^2 f_*(t) \sqrt{M}_*(t)\:dx_*dv_*\right|\leq C\eps^\eu
+\eps^{2(1-\eu)}\sup_{t\in[0,T_0]}\norm{\pi_L(f(t))}_{L^2_{x,v}}^2.
\end{equation}
The identities \eqref{Moment0OK}-\eqref{Moment1OKOK}-\eqref{Moment2OK} above imply
\begin{equation}\label{piGlobalt}
\sup_{t\in[0,T_0]}\int_{\R^d}\left|\int_{\T^d}\pi_L(f(t)) dx\right|^2 dv\leq C \eps^{2\eu}+2\eps^{2(1-\eu)}\sup_{t\in[0,T_0]}\norm{\pi_L(f(t))}_{L^2_{x,v}}^2.
\end{equation}
We combine \eqref{piGlobalt} with \eqref{eq:poincare0} to obtain
\[
\sup_{t\in[0,T_0]}\norm{\pi_L(f(t))}^2_{L^2_{x,v}} \leq C\norm{\nabla_x(\pi_L(f(t)))}^2_{L^2_{x,v}}+C\eps^{2\eu}+2\eps^{2(1-\eu)}\sup_{t\in[0,T_0]}\norm{\pi_L(f(t))}_{L^2_{x,v}}^2.
\]
For $\eps$ small enough, this gives
\begin{equation}\label{eq:control piL poincare}
\norm{\pi_L(f)}^2_{L^2_\gamma} \leq C_\pi C\norm{\nabla_x(\pi_L(f(t)))}^2_{L^2_{x,v}}+C_\pi C\eps^{2\eu}\leq C_\pi'\pa{\norm{\nabla_x(f(t))}^2_{L^2_{x,v}}+1},
\end{equation}
where $C_\pi$ was defined in $\eqref{eq:Cpi}$.

\bigskip
\textbf{Evolution of the $H^1$-norm.} Let us now look at the evolution of the full $H^1_{x,v}$-norm. We take $p$, $q$ and $r$ and define
$$Q_1(f) = p\norm{\nabla_x f}^2_{L^2_{x,v}} + q\eps^2\norm{\nabla_v f}^2_{L^2_{x,v}} +\eps r \langle \nabla_x f ,\nabla_v f \rangle_{L^2_{x,v}}.$$
Using the estimates given in Section \ref{sec:estimates operators} exactly as for \eqref{eq:Qli general} but keeping explicit the dependencies $C_\Lambda=\frac{C}{\Lambda}$ we find 
\begin{equation*}
\begin{split}
\frac{1}{2}\frac{d}{dt}Q_{1}(f) \leq & \frac{C_{\eta_0}r +  p\eta_0 - p}{\eps^2}\lambda_{0,\Lambda}\norm{\pi_L^\bot\pa{\nabla_x f}}^2_{L^2_\gamma} 
\\&+\cro{\frac{q\eps^2}{\eta_1}+\eta_4p+ \eps^{1-\eu}\frac{1_{\eps <1}}{4}p+\eps\frac{1_{\eps <1}}{4}r - r}\norm{\nabla_x f}^2_{L^2_{x,v}}
\\&+ \cro{q(\eta_1+\eta'_0+\eta_4+\eps^{3-e}\frac{1_{\eps <1}}{4})+r(\eta_0+\eta_4)\lambda_{0,\Lambda}+\eps^{2-e}\frac{1_{\eps <1}}{4}r - q\lambda_{s,\Lambda}} \norm{\nabla_v f}^2_{L^2_\gamma}
\\&+ \eps^{1+\eu}\Lambda\cro{C_{\eta_3}r + \eta_2 p-\frac{p}{2}}\norm{\nabla_x f}^2_{L^2_2}
\\&+ \eps^{3+e}\Lambda\cro{\eta_2 q + (\eta_2 + \eta_3) r  - \frac{q}{2}}\norm{\nabla_v f}^2_{L^2_2} 
\\&+ \eps^{1-\eu}\cro{pC_{\eta_2}+ d \frac{C_{\eta_2}}{\Lambda}r}\norm{\nabla_v f}^2_{L^2_{x,v}} 
\\&+\mathbf{\cro{q\frac{C_1C_\pi'}{\Lambda} + 3p\eps^{1-\eu}\frac{C_2C_\pi'}{\Lambda} + 2r\eps^{1-\eu}\frac{C_3C_\pi'}{\Lambda}}\norm{\nabla_x f}^2_{L^2_{x,v}}}
\\&+ C_{p,q,r,\Lambda,\eta}\cro{\norm{f}^2_{\mathcal{H}^{s}_{\eps}}\norm{f}^2_{H^{s}_\gamma} + \norm{\pi_L^\bot f}^2_{L^2_{x,v}}+ 1}.
\end{split}
\end{equation*}
We solely decomposed the $\norm{f}^2_{L^2_{x,v}}$ appearing in \eqref{eq:Qli general} into $\pi_L(f)$ - which we controlled thanks to \eqref{eq:control piL poincare} - and $\pi_L^\bot(f)$. To clarify we emphasized in bold the newly added terms.
For $\eps < 1$ we see that we can make exactly the same choices for the constants as in Proposition \ref{prop:estimate Qli} with the following two modifications which have no impact on the proof:
\begin{itemize}
\item{(8)} $p$ sufficiently large such that $C_{1,\eta_0}r + qC_{1} +p \eta_0 - p\lambda_{0,1} \leq -2$ - instead of $C_{1,\eta_0}r +p \eta_0 - p \leq -1$
\item{(10)} $\eps$ small enough as before plus $3p\eps^{1-\eu}\frac{C_2C_\pi'}{\Lambda} + 2r\eps^{1-\eu}\frac{C_3C_\pi'}{\lambda} \leq \lambda_{0,1}$.
\end{itemize}
In the case $\eps=1$ it is easier since the choices made in Proposition \ref{prop:estimate Qli} leave $\Lambda\geq 1$ free so we can fix all the constant in the same way - recall that $p$, $q$ and $r$ are independant of $\Lambda$ - and then choose $\Lambda$ sufficiently large such that
$$q\frac{C_1C_\pi'}{\Lambda} + 3p\eps^{1-\eu}\frac{C_2C_\pi'}{\Lambda} + 2r\eps^{1-\eu}\frac{C_3C_\pi'}{\Lambda} \leq -\frac{1}{2}\cro{ \frac{q\eps^2}{\eta_1}+\eta_4p+ \eps^{1-\eu}\frac{1_{\eps <1}}{4}p+\eps\frac{1_{\eps <1}}{4}r - r}.$$
In all cases we can find $p$, $q$ and $r$ such that $Q_1 \sim \norm{\nabla_x \cdot}^2_{L^2_{x,v}}+\eps^2\norm{\nabla_v \cdot}^2_{L^2_{x,v}}$ and
\begin{equation}\label{eq:control Q1}
\begin{split}
\frac{1}{2}\frac{d}{dt}Q_1(f) \leq& -\lambda\cro{\frac{1}{\eps^2}\norm{\pi_L^\bot(\nabla_xf)}^2_{L^2_\gamma} + \norm{\nabla_v f}^2_{L^2_\gamma}+\norm{\nabla_x f}^2_{L^2_{x,v}}} 
\\&+ C_s\norm{f}^2_{\mathcal{H}^s_\eps}\norm{f}^2_{\mathcal{H}^s_\gamma} + C_s\norm{\pi_L^\bot(f)}^2_{L^2_{x,v}} + C_s.
\end{split}
\end{equation}

\par At last, the evolution of the full $L^2_{x,v}$-norm is derived as before using the estimates given in Section \ref{sec:estimates operators}. One bounds
\begin{equation*}
\begin{split}
\frac{d}{dt}\norm{f}^2_{L^2_{x,v}} \leq& -\frac{\lambda_{0,\Lambda}}{\eps^2}\pa{1-2\eta_0}\norm{\pi_L^\bot(f)}^2_{L^2_\gamma} - \eps^{1+\eu}\Lambda\norm{f}^2_{L^2_2} + \eps^{1-\eu}\frac{\mathbf{1}_{\eps<1}}{4}\norm{f}^2_{L^2_{x,v}} 
\\&+ \lambda_{0,\Lambda} \norm{f}^2_{\mathcal{H}^{s}_{\eps}}\norm{f}^2_{H^s_\gamma} + 2C_{\Lambda,\eta_0}.
\end{split}
\end{equation*}
Once we fix $\eta_0$ sufficiently small and use the orthogonal projection together with the control of $\pi_L(f)$ by the spatial derivatives \eqref{eq:control piL poincare} the above implies
\begin{equation}\label{eq:control L2}
\begin{split}
\frac{d}{dt}\norm{f}^2_{L^2_{x,v}} \leq& -\cro{\frac{\lambda}{\eps^2}-\eps^{1-\eu}\frac{\mathbf{1}_{\eps<1}}{4}}\norm{\pi_L^\bot(f)}^2_{L^2_\gamma}- \eps^{1+\eu}\Lambda\norm{f}^2_{L^2_2} + \eps^{1-\eu}\frac{\mathbf{1}_{\eps<1}}{4}\norm{\nabla_x f}_{L^2_{x,v}}^2
\\&+ C_s \norm{f}^2_{\mathcal{H}^s_{\eps}}\norm{f}^2_{H^s_\gamma} + C_s.
\end{split}
\end{equation}
We add $\alpha\eqref{eq:control Q1}+\eqref{eq:control L2}$ with $\alpha$ sufficently small so that $\alpha C_s\norm{\pi_L^\bot(f)}^2_{L^2_{x,v}}$ from \eqref{eq:control Q1}is absorbed by $\frac{\lambda}{2\eps^2}\norm{\pi_L^\bot(f)}^2_{L^2_{x,v}}$ from \eqref{eq:control L2}. Then we can make $\eps=1$ or $\eps$ sufficiently small and obtain
\begin{equation}\label{eq:control H1}
\frac{1}{2}\frac{d}{dt}\pa{\norm{f}^2_{L^2_{x,v}}+\alpha Q_1(f)} \leq -\lambda'\norm{f}^2_{H^1_\gamma} + C_s\norm{f}^2_{\mathcal{H}^s_\eps}\norm{f}^2_{H^s_\gamma} + C_s.
\end{equation}

\bigskip
\textbf{Conclusion of the proof.} Taking a small parameter $\eta$, the linear combination $\eta\eqref{eq:start deterministic2}+\eqref{eq:control H1}$ shows the existence of $\lambda >0$ and $C_s>0$ such that
$$\frac{1}{2}\frac{d}{dt}\pa{\norm{f}^2_{L^2_{x,v}}+Q_1(f) + \eta\sum\limits_{s'=2}^sF_{s'}(f)}\leq -(\lambda-C_s\norm{f}^2_{\mathcal{H}^{s}_{\eps}})\norm{f}^2_{H^s_\gamma} + C_s$$
which concludes the proof by defining
\begin{equation}\label{eq:Hsxv}
\norm{f}^2_{\mathcal{H}^s_{x,v}} = \norm{f}^2_{L^2_{x,v}}+Q_1(f) + \eta\sum\limits_{s'=2}^sF_{s'}(f).
\end{equation}

\end{proof}

\begin{remark}\label{rem:e>0}
We see in the derivation of the estimate \eqref{GlobaldtMOK}, that we need $\eu>0$, to ensure that, although not conserved, the global quantities associated to the mass, momentum and energy of the perturbation $f$ are controlled in a suitable way.
\end{remark}


\section{Proofs of the main results}\label{sec:proof deterministic}

At last, we have all the tools to give the proof of the main results.


\subsection{Results on the Cauchy theory for the Boltzmann equation}

\begin{proof}[Proof of Theorem \ref{theo:Cauchy deterministic}]
The proof of existence follows from a standard iterative sche\-me:
\begin{multline*}
\partial_t f_{n+1} + \frac{1}{\eps}v\cdot\nabla_x f_{n+1} + \eps \vec{E_t}(x)\cdot\nabla_vf_{n+1} + \eps \mathcal{E}(t,x,v) f_{n+1} \\
= \frac{1}{\eps^2}L[f_{n+1}] + \frac{1}{\eps}\Gamma[f_n,f_{n+1}] 
-2 \mathcal{E}(t,x,v)M^{1/2}.
\end{multline*}
A detailed procedure is given in \cite[Section 6.1]{Bri3} for $\vec{E_t}=0$. This proof is directly applicable here, in combination with our estimates of $S_2$ and $S_3$ (terms involving $\vec{E_t}$). We obtain a uniform bound on a sequence of approximations $\pa{f_n}_{n\in\N}$ in $L^\infty_tH^s_{x,v} \cap L^1_t H^s_\gamma$, and therefore the strong convergence towards $f$ in less regular Sobolev spaces by Rellich's theorem. The uniqueness of the solution is also standard when $\vec{E_t}=0$. When $\vec{E_t}$ is non-trivial, uniqueness directly follows from our \textit{a priori} estimate method applied to the difference $f-g$ of two solutions: the linear parts are estimated in exactly the same way and the bilinear term is controledl when $\norm{f}^2_{H^s_{x,v}} + \norm{g}^2_{H^s_{x,v}}$ is small enough, which is why we obtain the uniqueness only in a perturbative regime.
\par We infer from our \textit{a priori} estimates described in Proposition \ref{prop:sobolev deterministic} that 
$$\forall t\in [0,T_0),\quad \frac{d}{dt}\norm{f}^2_{\mathcal{H}^s_{x,v}} \leq -(\lambda-C_s\norm{f}^2_{\mathcal{H}^s_{\eps}})\norm{f}^2_{H^s_\gamma} + C_s.$$
Coming back to the definition of the $\mathcal{H}^s_{\eps}$-norm given by $\eqref{eq:Sobolev eps}$ we see that it is uniformly equivalent to the $\mathcal{H}^s_{x,v}$-norm. Therefore we have
$$\forall t\in [0,T_0),\quad \frac{d}{dt}\norm{f}^2_{\mathcal{H}^s_{x,v}} \leq -(\lambda-C_s\norm{f}^2_{\mathcal{H}^s_{x,v}})\norm{f}^2_{H^s_\gamma} + C_s.$$
As $\norm{\cdot}_{\mathcal{H}^s_{x,v}} \leq \norm{\cdot}_{H^s_\gamma}$ we directly apply Gr\"onwall lemma which implies that 
$$\forall t\in [0,T_0),\quad\norm{f(t)}_{\mathcal{H}^s_{x,v}} \leq \max\br{\norm{f_\mathrm{in}}_{\mathcal{H}^s_{x,v}}, C_{T_0,E,s}}$$
as long as $\norm{f}^2_{\mathcal{H}^s_{x,v}}(t=0)$ is sufficiently small.

\begin{remark}\label{rem:time decay}
In the specific case where $\norm{\vec{E_t}} \leq \frac{C_E}{(1+t)^\alpha}$ with $\alpha >1$ the second constant $C_s$ behaves like $C_s/(1+t)^\alpha$ ($C_E$ is merely replaced by $C_E/(1+t)^\alpha$ in the computations). Gr\"onwall lemma then gives a polynomial time decay for $f$ for sufficiently small $\eps$.
\end{remark}
\end{proof}

The corollary follows at once.

\begin{proof}[Proof of Corollary \ref{cor:Cauchy deterministic}]
The corollary is a direct consequence of Theorem \ref{theo:Cauchy deterministic} and our definition of fluctuation $M$. Indeed, direct computations show that
$$\exists C_s >0,\:\forall t \geq 0,\quad \norm{\mu - M}_{H^s_{x,v}} \leq \eps^{1+\eu}C_s.$$
Therefore, using the notations of Theorem \ref{theo:Cauchy deterministic}
$$\pa{\norm{\frac{F_\mathrm{in}-\mu}{\eps\sqrt{\mu}}}_{H^s_{x,v}}\leq \frac{\delta_{T_0,E,s}}{2}} \Rightarrow \pa{\norm{\frac{F_\mathrm{in}-M}{\eps\sqrt{M}}}_{H^s_{x,v}} \leq \frac{\delta_{T_0,E,s}}{2} + \eps^\eu C_s}$$
which raises the expected corollary for $\eps$ sufficiently small because $\eu>0$: we construct a solution provided by Theorem \ref{theo:Cauchy deterministic} and this solution remains close to $\mu$ by an additive constant $\eps^\eu C_s$.
\end{proof}


\subsection{Results on the Hydrodynamical limit}

We are left with the computation of the limit equations and convergence issues.

\begin{proof}[Proof of Theorem \ref{theo:hydro lim}]
Let $T_0 >0$ and $F_\eps$ be the solution built in Corollary \ref{cor:Cauchy deterministic} on $[0,T_0]$ and define $f_\eps = \eps^{-1}\frac{F-\mu}{\sqrt{\mu}}$. The sequence $\pa{f_\eps}_{\eps>0}$ is uniformly bounded in $L^\infty_{[0,T_0]}H^s_{x,v}$ and solves
\begin{equation}\label{eq:BEeps}
BE_\eps(f_\eps) = \eps^2\mathcal{E}(t,x,v)\mu^{1/2}+\eps^3\cro{\vec{E_t}(x)\cdot\nabla_vf +\mathcal{E}(t,x,v) f}
\end{equation}
where $BE_\eps(f)$ is the standard Boltzmann equation operator given by
$$BE_\eps(f) = \eps^2\partial_t f + \eps v\cdot\nabla_x f -L[f] - \eps\Gamma[f,f]$$
and we recall that
$$\mathcal{E}(t,x,v) = \eps^\eu\frac{2A+a\abs{v}^2}{4(1+t)^{2}} - \frac{1}{2}\pa{1+\eps^{1+\eu}\frac{a}{1+t}} \vec{E_t}(x)\cdot v.$$
Having $\pa{f_\eps}_{\eps>0}$ uniformly bounded in $L^\infty_{[0,T_0]}H^s_{x,v}$ means that up to a subsequence $f_\eps$ converges weakly-* in this space towards $\bar{f}$. The choice of $s\geq s_0$ allows us to take the weak-* limit in $BE_\eps(f_\eps)$ as it is now standard in the field \cite{BarUka,StRay}, \cite[Section 8]{Bri3} and obtain
$$\lim\limits_{\eps\to 0} BE_\eps(\bar{f}) = L[\bar{f}].$$
The right-hand side of $\eqref{eq:BEeps}$ is linear in $f_\eps$ and it therefore converges weakly-* towards $0$. In the limit one must have
$$L[\bar{f}] = 0 \quad\mbox{so}\quad \bar{f}(t,x,v)= \cro{\rho(t,x) + v\cdot u(t,x) + \frac{\abs{v}^2-d}{2}\theta(t,x)}\sqrt{\mu(v)}.$$

The fluid equations are obtained by integrating in velocity $\eqref{eq:BEeps}$ against $\sqrt{\mu}$, $v\sqrt{\mu}$ and $\frac{\abs{v}^2-(d+2)}{2}\sqrt{\mu}$. The computations on the Boltzmann equation part $BE_\eps(f_\eps)$ have been done and proven rigorously for weaker convergences \cite{BGL,Gol14} and one obtain that
$$\lim\limits_{\eps\to 0 }\frac{1}{\eps}\int_{\R^d} \left(\begin{array}{c}1 \\ v \\ \frac{\abs{v}^2-(d+2)}{2} \end{array} \right)\sqrt{\mu}BE_\eps(f_\eps)dv = \pa{\begin{array}{c} \nabla_x\cdot u(t,x) \\ \nabla_x(\rho+\theta) \\ * \end{array}}$$
and so looking at the right-hand side of $\eqref{eq:BEeps}$ we see that in the limit
$$\nabla_x\cdot u(t,x)=0 \quad\mbox{and}\quad\nabla_x(\rho+\theta)=0$$
which are the incompressibility and Boussinesq relation.
\par Looking at the order $\eps^2$ in $BE_\eps(f_\eps)$ yields the Navier-Stokes-Fourier system in the Leray sense \cite{BGL,Gol14} - that is integrated against test functions with null divergence. It only remains to see what the right-hand side of $\eqref{eq:BEeps}$ becomes in the limit at order $\eps^2$:
\begin{equation*}
\begin{split}
\lim\limits_{\eps\to 0 }\frac{1}{\eps^2}&\int_{\R^d} \left(\begin{array}{c}1 \\ v \\ \frac{\abs{v}^2-(d+2)}{2} \end{array} \right)\cro{\eps^2\mathcal{E}(t,x,v)\mu+\eps^3\cro{\vec{E_t}(x)\cdot\nabla_vf +\mathcal{E}(t,x,v) f}\sqrt{\mu}}dv 
\\&= \lim\limits_{\eps\to 0 }\int_{\R^d} \left(\begin{array}{c}1 \\ v \\ \frac{\abs{v}^2-(d+2)}{2} \end{array} \right)\frac{\vec{E_t}(x)}{2}\cdot v\mu dv = \pa{\begin{array}{c} 0 \\ \frac{1}{2}\vec{E_t}(x) \\ 0 \end{array}}.
\end{split}
\end{equation*}
Therefore, taking the limit of the hydrodynamic quantities when $\eps$ goes to $0$ of $\eps^{-2}\eqref{eq:BEeps}$ yields that $(\rho,u,\theta)$ is a Leray solution to the incompressible Navier-Stokes equation with a force $\eqref{eq:NS force}$ together with the Boussinesq equation $\eqref{eq:Boussinesq}$.

\end{proof}

%
%

%
\bibliographystyle{acm}
\bibliography{bibliography_BE_INS_deterministic_force}

%
\bigskip
\signmarc
\signarnaud
\signjulien

\end{document}